%% file: main.tex
\documentclass{LMCS}

\def\dOi{9(4:6)2013}
\lmcsheading%
{\dOi}
{1--33}
{}
{}
{Mar.~28, 2011}
{Mar.~31, 2015}
{}

\usepackage{tikz}
\input{STYLES/STYLES}

\begin{document}

\title[Incarnation in Ludics]{Incarnation in Ludics and maximal cliques of paths\rsuper*}

\author[C.~Fouquer\'e]{Christophe Fouquer\'e\rsuper a}	
\address{{\lsuper a}Universit\'e Paris 13, Sorbonne Paris Cit\'e, LIPN, CNRS, (UMR 7030), F--93430, Villetaneuse, France}	
\email{christophe.fouquere@lipn.univ-paris13.fr}  

\author[M.~Quatrini]{Myriam Quatrini\rsuper b}	
\address{{\lsuper b}IML--FRE 3529 Aix-Marseille Universit\'e, CNRS, Campus de Luminy, case 907, F--13288 Marseille cedex 9, France}	
\email{quatrini@iml.univ-mrs.fr}

\ACMCCS{[{\bf Theory of computation}]:  Logic---Linear logic}

\thanks{\lsuper{a,b}This work has been partially funded by the French ANR projet blanc ``Locativity and Geometry of Interaction'' LOGOI ANR-10-BLAN-0213 02.}

\begin{revision}
  This is a revised and corrected version of the article originally
  published on October 16, 2013.
\end{revision}

\begin{abstract}
Ludics is a reconstruction of logic with interaction as a primitive notion, in the sense
that the primary logical concepts are no more formulas and proofs but cut-elimination
interpreted as an interaction between objects called designs.
When the interaction between two designs goes well, such two designs are said to be orthogonal.
A behaviour is a set of designs closed under bi-orthogonality.
Logical formulas are then denoted by behaviours. Finally proofs are interpreted as designs satisfying particular properties.
In that way, designs are more general than proofs and we may notice in particular that they are not typed objects.
{\em Incarnation} is introduced by Girard in Ludics as a characterization of ``useful'' designs in a behaviour. 
The incarnation of a design is defined as its subdesign that is the smallest one in the behaviour ordered by inclusion. 
It is useful in particular because being ``incarnated'' is one of the conditions for a design to denote a proof of a formula.
The computation of incarnation is important also as it gives a minimal denotation for a formula, and more generally for a behaviour. 
We give here a {\em constructive} way to capture the incarnation of the behaviour of a set of designs, without computing the behaviour itself. 
The method we follow uses an alternative definition of designs: rather than defining them as sets of chronicles, we consider them as sets of {\em paths}, a concept very close to that of play in game semantics that allows an easier handling of the interaction: the unfolding of interaction is a path common to two interacting designs. 

\end{abstract}

\maketitle\vfill

\section{Introduction}
\Rouge{
\subsection{Ludics: a theory of Logic based on interaction} The cut elimination procedure is one of the main facts of Proof Theory, as it corresponds to a notion of computation in computer science, the one used in functional programming. In Ludics, the logical theory that J.-Y. Girard developed in~\cite{DBLP:journals/mscs/Girard01}, the cut elimination moved from a central position towards a primitive one: technically, the {\it cut} is no more a formal rule, \ie, one of the elementary constructors of formal proofs; the cut-rule disappears, instead of which a notion of {\em interaction} arises which happens between cut-free proof-like objects called {\it designs}. A design is a forest with {\em actions} as nodes, where an action is an abstract view of an application of a rule. An interaction is a travel through two such forests. Two designs are said {\it orthogonal} when interaction terminates.

Such an approach resting on interaction has strong relations with the Geometry of Interaction project~\cite{DBLP:conf/colog/Girard88,towards} which aims at studying cut-elimination as an interaction between proofs seen as operators. As in Geometry of Interaction, the process of cut-elimination happens in Ludics between objects that are more general than proofs.%

Interaction between designs may also be understood as normalization in $\lambda$-calculus. As described by Terui~\cite{DBLP:journals/tcs/Terui11}, a design corresponds to a generalized untyped $\lambda$-term in which the abstraction/application duality appears as several $n$-ary abstraction/applications dualities.

\subsection{Ludics and Game Semantics} The approach developed in Ludics is closely related to another approach of calculus developed around a notion of interaction: the one occurring between player strategies in a game. Since the 90's,  beginning with seminal works of Lafont and Streicher~\cite{DBLP:conf/lics/LafontS91}, Game Semantics has been extremely fruitful for studying various fragments of Linear Logic or Classical Logic  in order to obtain 
full completeness results, from the multiplicative fragment~\cite{HylandOng93, DBLP:journals/jsyml/AbramskyJ94, DBLP:conf/lics/Loader94, DBLP:conf/lics/DevarajanHPP99,DBLP:conf/lics/AbramskyM99,DBLP:conf/lics/Lamarche95}, to Linear and Classical logics~\cite{Laurent04, Laurent05a,DalLagoLaurent08a,Laurent10}.  Considering event structures as a basis for Game Semantics, Mellies~\cite{DBLP:conf/lics/Mellies05, DBLP:conf/lics/MelliesT07} was able to prove full completeness for Linear Logic using {\em asynchronous} games: a formula is interpreted as an asynchronous game, and there is a full correspondence between proofs and strategies satisfying particular conditions.
It has 
been noticed that the basic concepts of Ludics may be expressed in terms of Game Semantics~\cite{DBLPconf/csl/Faggian02, DBLP:conf/csl/CurienF05,  DBLP:journals/corr/abs-cs-0501039} (see also~\cite{DBLP:journals/corr/abs-1104-0504} for a thorough presentation).
\begin{itemize}
\item an action is a {\em move}, the  abstraction of the application of a rule,
\item the sequence of actions used during interaction is a {\em play}, the cut-elimination steps,
\item a design is an {\em innocent strategy}, it is also a frame of a sequent calculus derivation.
\end{itemize}
However there is a fundamental difference between Game Semantics as it is generally used and Ludics.
Strategies are typed, while designs are {\it a priori} untyped.
More concretely, a game comes with a set of {\em plays}, \ie, sequences of {\em moves} that satisfy particular conditions, a {\em strategy} is nothing else but such a set of plays. In Ludics, a play is what results from the interaction between two designs, and 
 a game, what denotes a formula,
 is interpreted as a behaviour, \ie, a set of designs which is closed under bi-orthogonality. 
Notice that, considered as an element of a behaviour, only part of a design may be travelled during interactions with designs in the orthogonal of the behaviour, and this part has to be considered as a strategy in terms of Game Semantics.
In simple words, a behaviour constrains the use of a design, the design being viewed as a set of {\em potential} plays, hence only a {\em potential} strategy.

\subsection{Ludics: main results, new concepts} 

Starting from the sole notion of interaction, Ludics  defines new objects: 
designs which are completely defined by their interactions, as affirmed by the separation theorem~\cite{DBLP:journals/mscs/Girard01}. 
Then, once given objects of Ludics, it is possible to rebuild the logic. A formula is denoted by a {\em behaviour}; a proof of a formula is denoted by designs satisfying specific properties.
With this in mind, Ludics can be seen as a semantics of Linear Logic: (i) the meaning of a formula is a set of designs, (ii) a design may be viewed as abstracting a concrete, polarized and focalized proof (and taking into account infinity and failures). 
Besides the usual theorems (associativity, stability, $\dots$), two important properties are available within Ludics: internal and full completeness. Internal completeness results from the fact that the space of behaviours/formulas is fully describable in terms of multiplicative and additive connectives over behaviours/formulas.
Full completeness is obtained when one characterizes the correspondence between all proofs (of formulas) and specific designs (of behaviours\footnote{In fact behaviours together with a partial equivalence relation.}).
Ludics introduces also original additional concepts among which {\em locativity}, \ie, the replacement of a formula by its address. Making explicit locativity  plays an important role in recent versions of Geometry of Interaction~\cite{DBLP:journals/tcs/Girard11}, \cite{DBLP:journals/apal/Seiller12}. Another original new concept of Ludics is {\em incarnation}, a central concept in this paper.
Incarnation is a characterization of ``useful'' designs in a behaviour $\behaviour{E}$, \ie, designs necessary for defining $\behaviour{E}$ with respect to orthogonality. To be ``incarnated'', or ``material'', is one of the conditions for a design to denote a proof of a formula. The incarnation $\Bincarnation{\behaviour{E}}$ of a behaviour $\behaviour{E}$ is its set of incarnated designs, \ie, designs $\design{D}$ such that $\design{D} = \Dincarnation{D}{\behaviour{E}}$, the incarnation $\Dincarnation{D}{\behaviour{E}}$ of a design $\design{D}$ being defined as the smallest design included in $\design{D}$ and belonging to $\behaviour{E}$.
The computation of incarnation is important as it gives a kind of minimal presentation of behaviours, and in particular of the denotation of a formula or a type. 
A behaviour being defined as the closure of a set of designs with respect to orthogonality, it could be useful to have a {\em constructive} way to capture its incarnation, without computing the behaviour itself, particularly when considering not yet closed models (of programs, processes, ...).

\subsection{Our contribution}

Given a set $E$ of designs of the same base, not necessarily closed under bi-orthogonality, our aim is to compute the incarnation of the behaviour generated by this set, without having to fully determine this behaviour. 
Such a computation is difficult in general:
some pieces may be found in designs of the behaviour that are not already present in $E$, and the computation of the ones that are relevant with respect to incarnation does not seem simpler than computing the behaviour itself. 
The method we follow uses an alternative definition of designs: rather than using their original definition as sets of {\it chronicles}~\cite{DBLP:journals/mscs/Girard01}, 
we consider them as sets of {\em paths}, a concept very close to that of play in game semantics, that allows an easier handling of the interaction: the unfolding of an interaction is a path common to two interacting designs. 
Among the paths of a set of designs $E$, we characterize those that are really visited during the interaction with $E^\perp$, the orthogonal of $E$. 
Hence a two-step process may be followed:
designs in the incarnation of $E^\perp$ are built from particular maximal cliques of visitable paths and, these designs being computed, it remains to repeat the same operation for getting the incarnation of the behaviour generated by the set $E$.

The rest of this article is organized as follows:
\begin{itemize}
\item In section~\ref{sec:Ludics}, we recall the original definitions of Ludics objects, in particular designs as cliques of chronicles and the interaction process.
\item In section~\ref{sec:Paths}, we define a notion of path together with a definition of coherence on paths. We present an equivalent definition of a design as a clique of paths. 
\item We characterize in section~\ref{section:normalization_visitable} those paths that are visitable by interaction.
\item In section~\ref{sec:Incarnation}, we focus on incarnation. We show that incarnation of the dual is obtained from particular maximal cliques of visitable paths, from which we obtain a double-step process for computing the incarnation of the behaviour generated by a set of designs. 
\end{itemize}
}
\section{Basic facts about Ludics}\label{sec:Ludics}
\subsection{Ludics: an informal introduction from a proof theoretical point of view} 
Ludics~\cite{DBLP:journals/mscs/Girard01} is a reconstruction of logic with {\em interaction} as a primitive notion, in the sense that the primary logical concepts  are no more formulas and proofs but cut-elimination interpreted as an interaction between objects called {\em designs}. 
The structure of a design may be introduced by considering the following questions concerning proof structure and cut-elimination: what is necessary to know about a proof when we are interested in its behaviour during the process of cut-elimination? Can such a proof-object be completely defined by this behaviour? 
Notice that the cut elimination process consists in travelling
along cut-proofs of formulas by passing through each of their rules from the conclusion of the rule towards
their premisses. Therefore, determining each possible cut elimination process for a given formula is closely related to the search of all possible proofs of this formula. Hence,
 the starting question may be formulated as follows: what are the relevant parts of a proof-object with respect to proof search?

\newcommand{\xlong}{2.5}
\newcommand{\ylong}{.8}
\newcommand{\along}{.2}

\begin{figure}
\begin{center}
\begin{tikzpicture}
\filldraw[red!20] (0,0) rectangle (\xlong,\ylong);
\filldraw[blue!20] (\xlong,0) rectangle (2*\xlong,\ylong);

\filldraw[blue!20] (0,\ylong) rectangle (\xlong,2*\ylong);
\filldraw[red!20] (\xlong,\ylong) rectangle (2*\xlong,2*\ylong);

\draw (0,2*\ylong) rectangle (\xlong,3*\ylong);
\draw (\xlong,2*\ylong) rectangle (2*\xlong,3*\ylong);

\node at (\xlong/2,\ylong/2) {$\parr$};
\node at (3*\xlong/2,\ylong/2) {$\oplus$};
\node at (\xlong/2,3*\ylong/2) {$\otimes$};
\node at (3*\xlong/2,3*\ylong/2) {$\with$};
\node at (\xlong/2,5*\ylong/2) {multiplicative};
\node at (3*\xlong/2,5*\ylong/2) {additive};

\draw (2*\xlong,0) rectangle (3*\xlong,\ylong);
\draw (2*\xlong,\ylong) rectangle (3*\xlong,2*\ylong);
\node at (5*\xlong/2,\ylong/2) {{\em or}};
\node at (5*\xlong/2,3*\ylong/2) {{\em and}};

\draw (-\xlong,0) rectangle (0,\ylong);
\draw (-\xlong,\ylong) rectangle (0,2*\ylong);

\filldraw[blue!20] (-\xlong,2*\ylong) rectangle (0,3*\ylong);
\filldraw[red!20] (2*\xlong,2*\ylong) rectangle (3*\xlong,3*\ylong);
\node at (-\xlong/2,5*\ylong/2) {{\bf positive}};
\node at (5*\xlong/2,5*\ylong/2) {{\bf negative}};

\draw[line width=1pt,->] (\xlong/2,5*\ylong/2-\along) -- (\xlong/2,5*\ylong/2-2*\along);
\draw[line width=1pt,->] (3*\xlong/2,5*\ylong/2-\along) -- (3*\xlong/2,5*\ylong/2-2*\along);

\draw[line width=1pt,->] (5*\xlong/2-2*\along,3*\ylong/2) -- (5*\xlong/2-3*\along,3*\ylong/2);
\draw[line width=1pt,->] (5*\xlong/2-2*\along,\ylong/2) -- (5*\xlong/2-3*\along,\ylong/2);

\draw[line width=1pt,->] (0-\along,5*\ylong/2-\along) -- (0,5*\ylong/2-2*\along);
\draw[line width=1pt,->] (2*\xlong+\along,5*\ylong/2-\along) -- (2*\xlong,5*\ylong/2-2*\along);

\end{tikzpicture}
\end{center}
\caption{Polarity of connectives}
\label{fig:Polarities}
\end{figure}

As studied by Andreoli~\cite{DBLP:journals/logcom/Andreoli92}, Linear Logic sequent calculus admits a focalized presentation:
the proof search may be viewed as a recursive process alternating reversible (or asynchronous, negative) and non-reversible (or synchronous, positive) steps (see Fig.~\ref{fig:Polarities}). 
In such a bottom-up approach, as soon as there is a negative formula, \ie, a formula having as principal connective $\parr$ or $\with$ in the conclusion, such a negative formula is decomposed until its positive subformulas, \ie, formulas with 
$\otimes$ or $\oplus$ as principal connective. Otherwise, if there are only positive formulas in the conclusion, one of these positive formulas is decomposed until its negative subformulas. As proved by Andreoli~\cite{DBLP:journals/logcom/Andreoli92}, this focusing discipline is complete with respect to provability.
Going one step further, a {\em hypersequentialized} version of multiplicative-additive Linear Logic sequent calculus may be defined~\cite{DBLP:journals/logcom/Andreoli92,girard00,DBLP:journals/iandc/CurienF12} with only  two kinds of logical rules (the positive and the negative) beside the axioms, one strictly alternating with the other.
Formulas  and rules of the hypersequentialized sequent calculus are given in Fig.~\ref{fig:HScalculus}.

\begin{figure}
\renewcommand{\tabularxcolumn}[1]{>{\centering\arraybackslash}m{#1}}
\begin{tabularx}{\textwidth}{lXXX}
{\bf Formulas} & 
\multicolumn{3}{c}{$
P ~|~ \un ~|~ \zero ~|~ (F^\perp\otimes\dots\otimes F^\perp)\oplus\dots\oplus (F^\perp\otimes\dots\otimes F^\perp)
$}
 \\
~\\
{\bf Sequents} & 
\multicolumn{3}{c}{
either $\vdash \Delta$ or $F \vdash \Delta$ where $F$ is a formula, $\Delta$ is a sequence of formulas
}
 \\
~
\\
{\bf Axiom rules}
&
$
\infer
	{
	P \vdash P,\Delta
	}
	{
	}
$
&
$
\infer {\vdash \un}{}
$
&
$
\infer {\zero\vdash \Delta}{}
$
\\
~
\\
{\bf Cut rule}
&
\multicolumn{3}{c}{$
\infer
	{
	\Gamma \vdash \Delta_1,\Delta_2
	}
	{
	\Gamma \vdash A,\Delta_1
	&
	A \vdash \Delta_2
	}
$}
\\
~
\\
{\bf Negative rule}
&
\multicolumn{3}{c}{
$
\infer
	{
	(A_{11}^\perp\otimes \dots \otimes A_{1n_1}^\perp)\oplus\dots\oplus(A_{p1}^\perp\otimes\dots\otimes A_{pn_p}^\perp)\vdash\Delta
	}
	{
	\vdash A_{11},\dots,A_{1n_1},\Delta
	&
	\dots
	&
	\vdash A_{p1},\dots,A_{pn_p}, \Delta
	}
$
}
\\
~
\\
{\bf Positive rule}
&
\multicolumn{3}{c}{
$
\infer
	{
	\vdash(A_{11}^\perp\otimes \dots \otimes A_{1n_1}^\perp)\oplus\dots\oplus(A_{p1}^\perp\otimes\dots\otimes A_{pn_p}^\perp), \Delta
	}
	{
	A_{i 1}\vdash\Delta_1
	&
	\dots
	&
	A_{in_i}\vdash \Delta_{n_i}
	}
$
}
\\
&
\multicolumn{3}{c}{
where $\cup \Delta_k\subset\Delta$ and for $k,l\in\{1,\dots n_i\}$, $\Delta_k\cap\Delta_l=\emptyset$.
}
\end{tabularx}
\caption{Hypersequentialized calculus}
\label{fig:HScalculus}
\end{figure}

 We give below examples of part of proofs in the hypersequentialized calculus. 

\begin{exa}
Let $A$, $B$, $C$, $D$, $E$ be negative formulas and $F$ be a positive one, formulas $A\otimes B\otimes C$ and $A\otimes ((D \otimes E)\parr F)$ are positive. Their proofs in the hypersequentialized calculus begin in the following way:
$$
\infer
	{
	\vdash A\otimes B\otimes C
	}
	{
	\deduce{A^\perp \vdash}{\vdots}
	&
	\deduce{B^\perp \vdash}{\vdots}
	&
	\deduce{C^\perp \vdash}{\vdots}
	}
\hspace{2cm}
\infer
	{
	\vdash A\otimes ((D \otimes E) \parr F)
	}
	{
	\deduce{A^\perp \vdash}{\vdots}
	&
	\infer
		{
		((D \otimes E) \parr F)^\perp \vdash
		}
			{\infer{\vdash D \otimes E, F}{
				\deduce{D^\perp \vdash F}{\vdots}
				&
				\deduce{E^\perp \vdash}{\vdots}
			}
			}
	}
$$
The proof trees may also be represented by means of derivation trees where nodes are the rules involved in the proofs, noting $+$ when it is a positive rule, $-$  otherwise:
\begin{center}
\begin{tikzpicture}
\node at (0,0) {$(+,A\otimes B\otimes C,\{A^\perp,B^\perp,C^\perp\})$}; 
\node at (0,.6) {$\vdots$};
\end{tikzpicture}
\hspace{.1cm}
\begin{tikzpicture}
\node at (0,0) {$(+,A\otimes  ((D \otimes E) \parr F),\{A^\perp,((D\otimes E) \parr F)^\perp\})$};
\node at (-2,1.2) {$\vdots$};
\node at (2,1) {$(-, ((D\otimes E) \parr F)^\perp,\{D \otimes E, F\})$};
\node at (2,2) {$(+,D \otimes E , \{D,E\})$};
\node at (1,3.2) {$\vdots$};
\node at (3,3.2) {$\vdots$};
\draw (0,.3) -- (-2,.7);
\draw (0,.3) -- (2,.7);
\draw (2,1.3) -- (2,1.7);
\draw (2,2.3) -- (1,2.7);
\draw (2,2.3) -- (3,2.7);
\end{tikzpicture}
\end{center}
\label{exa:proofs}
\end{exa}

\noindent Let us outline that such a hypersequentialized presentation of proof objects is relevant for  our starting question. It enables one to get rid of irrelevant  parts of proofs with respect to proof-search or equivalently  cut elimination. For example, suppose that $A$, $B$ and $C$ are negative linear formulas and that we are looking for a proof of $A \otimes (B \oplus C)$ or for a proof of $(A \otimes B) \oplus (A \otimes C)$. We just need to find either some proofs of $A$ and $B$ or some proofs of $A$ and $C$. Equivalently in terms of cut-elimination, starting from a cut on $A \otimes (B \oplus C)$  or on  $(A \otimes B) \oplus (A \otimes C)$, the process continues either with cuts on $A$ and on $B$, or with cuts on $A$ and on $C$. Hence the distinction between formulas $A \otimes (B \oplus C)$ and  $(A \otimes B) \oplus (A \otimes C)$ is not relevant. Furthermore, it is useless to consider subformulas  $B \oplus C$, $A \otimes B$, $A \otimes C$. What is relevant is to consider subformulas with opposite polarity, here $A$, $B$ and $C$. This is realized in hypersequentialized discipline, where we keep only one of the two equivalent formulas, \eg\ $(A \otimes B) \oplus (A \otimes C)$, and where the logical rule enables to go directly to subformulas $A$, $B$ and $C$.

A still more radical choice is then made in Ludics: instead of (canonical) formulas, only addresses are kept. A subformula has an address relatively to the occurrence of a formula in which it appears. Such an address is the {\em locus} on which an {\em interaction}, \ie, a cut, may take place. In terms of Ludics, a rule that decomposes a formula into subformulas is replaced by a rule applied to a locus that generates a finite number of subloci. 
A locus is presented as a finite sequence of integers. Subloci generated by applying a rule to a locus $\xi$ are built by increasing the sequence $\xi$ with arbitrary distinct integers. 

\begin{exa}
The two pieces of proofs in (classic) Linear Logic sequent calculus:\\
\begin{center}
$\quad$  \shortstack{
 \shortstack{ 
 $\vdots$\\
                     $\vdash A$}
                     \hspace{1em}
  \shortstack{
                       $\vdots$\\
                     $\vdash B$} \\
                     $\hrulefill$\\
                     $\vdash A\otimes B$\\
                     $\hrulefill$\\
                     $\vdash (A\otimes B)\oplus(A\otimes C)$}      
                     $\quad$ and $\quad$
      \shortstack{
 \shortstack{
 $\vdots$\\
                     $\vdash A$}
                     \hspace{1em}
  \shortstack{
                       $\vdots$\\
                     $\vdash B$\\
                     $\hrulefill$\\
                     $\vdash B\oplus C$} \\
                     $\hrulefill$\\
                     $\vdash A\otimes(B\oplus C)$}
\end{center}

\noindent are both replaced in the hypersequentialized calculus by: $\quad$
                              \shortstack{
                     \shortstack{
                     $\vdots$\\
                     $A^\perp\vdash $}\hspace{2em}
                     \shortstack{
                     $\vdots$\\
                     $B^\perp\vdash $}\\
                     $\hrulefill$\\
                     $\vdash (A\otimes B)\oplus(A\otimes C)$}    \\
and are both replaced in Ludics by: $\quad$
                              \shortstack{
                     \shortstack{
                     $\vdots$\\
                     $\xi.1\vdash $}\hspace{2em}
                     \shortstack{
                     $\vdots$\\
                     $\xi.2\vdash $}\\
                     $\hrulefill$\\
                     $\vdash\xi$}

\label{loci}
\end{exa}
In the hypersequentialized calculus, one might expect that a proof search stops once it  arrives at an axiom. However in Ludics there are no more formulas, thus no more axioms. Obviously, there should exist a way to stop an interaction. Still using the proof search intuition, an interaction ends when one of the interacting designs gives up: this is expressed by means of a special rule in Ludics, called {\it daimon}.

Once the above generalization is performed, presenting designs as abstract proofs may remain ambiguous.
In fact, as it is the case in example~\ref{dessin}, two proof-like presentations of designs that only differ on the distribution of weakened contexts cannot be distinguished by interaction: it is not possible to find a counter-proof with which the result of the cut-elimination procedure would be distinct. 
In terms of Ludics, a rule is presented as an {\em action} of the form $(+/-,\xi,\{i_1,\dots,i_n\})$ where $i_j$ are integers concatenated as suffixes to $\xi$ in order to create subaddresses,
and designs are presented as trees (or more likely forests) where nodes are actions.

\begin{exa}
The two (complete) proof-like trees on the left are more likely represented by the tree of  actions on the right:
\begin{center}
$
\infer{\xi\vdash}
	{
	\infer{\vdash \xi.1, \xi.2, \xi.3}
		{
		\xi.1.4 \vdash \xi.2
		&
		\xi.1.7 \vdash \xi.3
		}
	}
$
and
$
\infer{\xi\vdash}
	{
	\infer{\vdash \xi.1, \xi.2, \xi.3}
		{
		\xi.1.4 \vdash \xi.2, \xi.3
		&
		\xi.1.7 \vdash
		}
	}
$
\hspace{.3cm}by\hspace{.3cm}
\raisebox{-.6em}{
\begin{tikzpicture}
\node at (0,0) {$(-,\xi,\{1,2,3\})$};
\node at (0,1) {$(+,\xi.1,\{4,7\})$}; 
\draw (0,.3) -- (0,.7);
\end{tikzpicture}
}
\end{center}
The object in Ludics is the following set of chronicles:
$$
\{(-,\xi,\{1,2,3\}) ~;
$$
$$
(-,\xi,\{1,2,3\})(+,\xi.1,\{4,7\})\}
$$
\label{dessin}
\end{exa}

\subsection{Designs as cliques of chronicles}
 The main object of Ludics, the {\bf design}, is defined by Girard~\cite{DBLP:journals/mscs/Girard01} as a set of pairwise coherent {\bf chronicles}, \ie\ a clique of chronicles, where chronicles are alternate sequences of {\bf actions}, the coherence relation being defined below. The actions themselves are built using a notion of {\bf locus}.
 
  \begin{defi}[Locus]
 A locus is a finite (maybe empty) sequence of integers.
  \end{defi}
\noindent{\sc Notation:} The loci  will be denoted by Greek letters: $\xi$, $\sigma$, \dots If $I$ is a finite set of integers, we denote $\xi.I$ the set $\{ \xi.i\;;\;i\in I\}$. 

 \begin{defi}[Action]
An {\bf action} $\kappa$ is
\begin{itemize}
\item either a positive proper action $(+,\xi,I)$ or a negative proper action $(-,\xi,I)$ where the locus $\xi$ is said the {\em focus} of the action, and the finite set of integers $I$ is said its {\em ramification},

\item or the positive action daimon denoted by $\daimon$.
\end{itemize}
 A locus $\xi.i$ is {\em justified} by an action $(+,\xi,I)$ when $i \in I$. By extension an action $(\epsilon,\xi.i,J)$ is {\em justified} by an action $(\overline{\epsilon},\xi,I)$ when $i \in I$, $\epsilon \in \{+,-\}$, $\overline{+} = -$ and $\overline{-} = +$. 
\end{defi}

In~\cite{DBLP:journals/mscs/Girard01} chronicles are defined as follows:

\begin{defi}[Chronicle]
A {\bf chronicle} $\chronicle{c}$ is a non-empty and finite alternate sequence of actions such that
\begin{itemize}
\item {\em Positive proper action:} A  positive proper action is either justified, \ie, its focus is built by one of the previous actions in the sequence, or it is called initial. 
\item {\em Negative action:} A negative action may be initial, in such a case it is the first action of the chronicle. Otherwise it is justified by the immediate previous positive action.
\item {\em Linearity:} Actions have distinct focuses.
\item {\em Daimon:} If present, a daimon ends the chronicle. 
\end{itemize}
\end{defi}

\begin{defi}[Coherence on Chronicles]
Two chronicles $\chronicle{c_1}$ and $\chronicle{c_2}$ are {\bf coherent}, noted $\chronicle{c_1} \coh_c \chronicle{c_2}$, when the two following conditions are satisfied:
\begin{itemize}
\item {\em Comparability:}  Either one extends the other or they first differ on negative actions, \ie, if $w\kappa_1 \coh_c w\kappa_2$ then either $\kappa_1 = \kappa_2$ or $\kappa_1$ and $\kappa_2$ are negative actions.
\item {\em Propagation:} When they first differ on negative actions and these negative actions have distinct focuses then the focuses of following actions in $\chronicle{c_1}$ and $\chronicle{c_2}$ are pairwise distinct, \ie, if $w(-,\xi_1,I_1)w_1\kappa_1 \coh_c w(-,\xi_2,I_2)w_2\kappa_2$ with $\xi_1 \neq \xi_2$ then $\kappa_1$ and $\kappa_2$ have distinct focuses.
\end{itemize}
\end{defi}

\noindent{\sc Notation:} An action is denoted $\alpha$, $\kappa$ or even $\kappa^{\epsilon}$ where $\epsilon$ is equal to $+$ or $-$, when we want to precise its polarity. $w$ denotes a sequence of actions. A chronicle is denoted by $\chronicle{c}$, $\chronicle{d}$, $\dots$  or by $\kappa_0\dots\kappa_n$ when we need to precise the actions occurring in the chronicle.\\

\begin{exa}~\\
\begin{minipage}{.66\textwidth}
To summarize, a chronicle is an alternate sequence of actions. A set of coherent chronicles may be presented as a tree (in fact, a forest when the first actions of chronicles are negative). Such a tree branches at negative nodes only and a branch is a chronicle. The tree on the right is a set of two coherent chronicles.
\end{minipage}
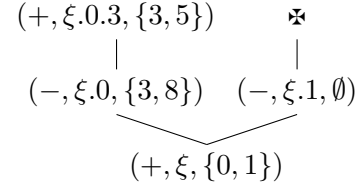
\begin{minipage}{.3\textwidth}
\raisebox{-.6em}{
\begin{tikzpicture}
\node at (0,0) {$(+,\xi,\{0,1\})$};
\node at (-1.2,1) {$(-,\xi.0,\{3,8\})$};
\node at (-1.2,2) {$(+,\xi.0.3,\{3,5\})$};
\node at (1.2,1) {$(-,\xi.1,\emptyset)$};
\node at (1.2,2) {$\daimon$};

\draw (0,.3) -- (-1.2,.7);
\draw (-1.2,1.3) -- (-1.2,1.7);
\draw (0,.3) -- (1.2,.7);
\draw (1.2,1.3) -- (1.2,1.7);
\end{tikzpicture}
}
\end{minipage}
\end{exa}

In order to be able to manipulate several chronicles together, it is convenient to consider chronicles with a given {\bf base}, \ie, a sequent of loci written $\Gamma\vdash\Delta$ such that $\Delta$ is a finite set of  loci and $\Gamma$ contains at most one locus and the loci belonging to $\Gamma\cup\Delta$ are pairwise disjoint, \ie, no locus is a sublocus of another one. When $\Gamma$ is empty, the base is called positive, otherwise it is called negative.
A chronicle $\chronicle{c}$ is said {\bf based on} $\Gamma\vdash\Delta$ provided that $\Gamma$ contains the focus of the initial negative action of $\chronicle c$ if it exists (otherwise $\Gamma$ is empty) and  $\Delta$ contains the focuses of initial positive actions (so a finite number).
Note that coherence of chronicles does not rely on their bases. In fact,  the definition of the base of a chronicle is liberal: there is not a unique base for each chronicle but several ones. Suppose $\Gamma\vdash\Delta$ is a base for a chronicle $\chronicle c$, and that $\Delta\subset\Delta'$, then $\Gamma\vdash\Delta'$ is also a base of $\chronicle c$.
We can observe that $\Gamma\cup\Delta$ must contain at least all the focuses of initial actions of $\chronicle c$ for $\Gamma\vdash\Delta$ to be a base of $\chronicle c$. 
Nevertheless one may consider the coherence of two chronicles either when the base of the first one is included in the base of the second one, or when the two bases are disjoint. In this latter case, the only possibility of coherence is when the bases are negative ones.

\begin{defi}[Designs, Slices, Nets]~
  \begin{itemize}
\item[$\bullet$] A {\bf design} $\design{D}$, based on
  $\Gamma\vdash\Delta$, is a set of chronicles based on
  $\Gamma\vdash\Delta$, such that the following conditions are
  satisfied:
\begin{itemize}[label=$-$]
\item  {\em Forest:} The set of chronicles is prefix closed.
\item  {\em Coherence:} The set is a clique of chronicles with respect to $\coh_c$. 
\item {\em Positivity:} A chronicle without extension in $\design{D}$ ends with a positive action.
\item {\em Totality:} $\design{D}$ is non-empty when the base is positive, in that case all the chronicles begin with a (unique) positive action.
\end{itemize}
\item[$\bullet$]  A {\bf slice} is a design $\design{S}$ such that if $w(-,\xi,I_1), w(-,\xi,I_2) \in \design{S}$ then $I_1 = I_2$. 
\item[$\bullet$]  A {\bf net} is a finite set of designs on disjoint bases. 
\end{itemize}
\end{defi}

A design may contain several chronicles ending with a negative action of same focus: this allows for representing the `with' `$\with$' additive connective of Linear Logic. A design may be split with respect to such situations to give rise to {\em slices}: a slice has a multiplicative structure.

\begin{exa}
The following example of design (design $\design{D}$) is based on $\vdash \xi$ and gives rise to two slices $\design{E}$ and $\design{F}$ (figures below).
 
\begin{center}
Design $\design{D} =$ \hspace{-4cm}
\raisebox{-.6em}{
\begin{tikzpicture}
\node at (0,0) {$(+,\xi,\{0,1\})$};
\node at (0,1) {$(-,\xi.0,\{4\})$};
\node at (0,2) {$(+,\xi.0.4,\{4\})$};
\node at (-4,1) {$(-,\xi.0,\{3,8\})$};
\node at (-4,2) {$(+,\xi.0.3,\{3,5\})$};
\node at (3,1) {$(-,\xi.1,\emptyset)$};
\node at (3,2) {$\daimon$};

\draw (0,.3) -- (0,.7);
\draw (0,1.3) -- (0,1.7);
\draw (0,.3) -- (-4,.7);
\draw (-4,1.3) -- (-4,1.7);
\draw (0,.3) -- (3,.7);
\draw (3,1.3) -- (3,1.7);
\end{tikzpicture}
}
\end{center}

\vspace{.3cm}

\begin{center}
Slice $\design{E} =$ \hspace{-.9cm}
\raisebox{-.6em}{
\begin{tikzpicture}
\node at (0,0) {$(+,\xi,\{0,1\})$};
\node at (-1.2,1) {$(-,\xi.0,\{3,8\})$};
\node at (-1.2,2) {$(+,\xi.0.3,\{3,5\})$};
\node at (1.2,1) {$(-,\xi.1,\emptyset)$};
\node at (1.2,2) {$\daimon$};

\draw (0,.3) -- (-1.2,.7);
\draw (-1.2,1.3) -- (-1.2,1.7);
\draw (0,.3) -- (1.2,.7);
\draw (1.2,1.3) -- (1.2,1.7);
\end{tikzpicture}
}
\quad\quad
Slice $\design{F} =$ \hspace{0cm}
\raisebox{-.6em}{
\begin{tikzpicture}
\node at (0,0) {$(+,\xi,\{0,1\})$};
\node at (0,1) {$(-,\xi.0,\{4\})$};
\node at (0,2) {$(+,\xi.0.4,\{4\})$};
\node at (2,1) {$(-,\xi.1,\emptyset)$};
\node at (2,2) {$\daimon$};

\draw (0,.3) -- (0,.7);
\draw (0,1.3) -- (0,1.7);
\draw (0,.3) -- (2,.7);
\draw (2,1.3) -- (2,1.7);
\end{tikzpicture}
}
\end{center}
 
\end{exa}

\noindent A proof-like presentation of a design is given by the following notion of {\bf design as a dessin}:\medskip

\begin{defi}[Design as dessin]~
A design as dessin based on $\Gamma\vdash\Delta$ is a tree of bases with root $\Gamma \vdash \Delta$ and built by means of the following rules: 
\begin{itemize}[label=$-$]
\item {\sc Daimon}
$$
\infer[\daimon]{\vdash\Delta}{}
$$
\item {\sc Positive rule}
$$
\infer[{(+,\xi,I)}]{\vdash \Delta, \xi}{\dots & \xi.i \vdash \Delta_i & \dots}
$$
for $ i\in I$, the $\Delta_i$ are pairwise disjoint and included in     $\Delta$.
\item {\sc Negative rule} 
$$
\infer[(-,\xi, \mathfrak N)]{\xi \vdash \Delta}{\dots & \vdash \xi.I, \Delta_I & \dots}
$$
${\mathfrak  N}$ is a set  (maybe empty or infinite) of ramifications. For all  $I\in {\mathfrak  N}$,  the $\Delta_I$'s are not necessarily disjoint and are included in  $\Delta$. 
\end{itemize} 

\end{defi}

\noindent With a design based on $\Gamma\vdash\Delta$, may be associated one or several proof-like presentations: a design as dessin of same base, with a positive rule for each positive action and a negative rule containing all the negative actions with the same focus. Nevertheless, such a presentation is not unique, as it is illustrated in the following example. 
\begin{exa} {\bf Designs and their proof-like presentation as dessins}\label{design-as-dessin}
\begin{itemize}
\item The design based on  $\vdash\Gamma$ which contains a unique chronicle, itself reduced to the action  $\daimon$ is denoted $\dai_+$. It has a unique proof-like presentation: 
$$
\infer[\daimon]{\vdash~ \Gamma}{}
$$
\item The empty design based on  $\xi\vdash$ is denoted by ${\design S}k_\xi$. It has also a unique proof-like presentation: the tree reduced to the root  $\xi\vdash$.
\item  The following design based on $\xi\vdash$:
 $$
 \begin{array}{cl}
 \design D=\{ &  (-,\xi,\{1,2,3\});\\
              & (-,\xi,\{1,2,3\})(+,\xi.1,\{4,7\})\}
                        \end{array}
                        $$
                        has several proof-like presentations, for example the two following ones:
\begin{center}
   \shortstack{
 \shortstack{  $\xi.1.4\vdash \xi.2$}
                     \hspace{1em}
  \shortstack{  $\xi.1.7\vdash \xi.3$} \\
                     $\hrulefill$\\
                     $\vdash \xi.1,\xi.2,\xi.3$\\
                     $\hrulefill$\\
                     $\xi\vdash$}      
                     $\quad$ or $\quad$
                     \shortstack{
     \shortstack{   $\xi.1.4\vdash \xi.2,\xi.3$}
                     \hspace{1em}
  \shortstack{  $\xi.1.7\vdash  $} \\
                     $\hrulefill$\\
                     $\vdash \xi.1,\xi.2,\xi.3$\\
                     $\hrulefill$\\
                     $\xi\vdash$}                             
\end{center}
 \item  	The following design is based on $\vdash\xi$:
 $$
 \begin{array}{cl}
 \design E=\{ &  (+,\xi,\{1,3\});\\
              & (+,\xi,\{1,3\})(-,\xi.1,\{0\});\\
                   & (+,\xi,\{1,3\})(-,\xi.1,\{0\})(+,\xi.1.0,\{0\});\\
                   &  (+,\xi,\{1,3\})(-,\xi.1,\{1\});\\
                    &  (+,\xi,\{1,3\})(-,\xi.1,\{1\})(+,\xi.1.1,\{0\})\\
                    &  (+,\xi,\{1,3\})(-,\xi.3,\{0\});\\
                    &  (+,\xi,\{1,3\})(-,\xi.3,\{0\})(+,\xi.3.0,\emptyset)\}
                        \end{array}
 $$
 Its unique proof-like presentation is:
$$
	\infer{\vdash \xi}
{
		\infer{\xi.1 \vdash}
			{
			\infer{\vdash \xi.1.0}
				{
				\xi.1.0.0 \vdash
				}
				&
				\infer{\vdash \xi.1.1}
				{
				\xi.1.1.0 \vdash
				}
			}
		&
		\infer{\xi.3 \vdash}
			{
			\infer{\vdash \xi.3.0}
				{
				\vdash
				}
			}
		}
	$$	
\end{itemize}
\end{exa}

\subsection{Interaction and Behaviours}

Interaction, \ie, cut elimination, is a normalization of particular nets of designs, called {\it cut-nets}~\cite{DBLP:journals/mscs/Girard01}. We give below the definition of interaction in the case of a closed cut-net, \ie, all addresses in bases are part of a cut. Therefore, one main design is distinguished because its base has no left part, \ie, is positive. 
The reader may find in~\cite{DBLP:journals/mscs/Girard01} the definition for the general case.

\begin{defi}[Closed cut-net]
Let $\design{R}$ be a net of designs, $\design{R}$ is a {\em closed cut-net} if 
\begin{itemize}
\item addresses in bases are either distinct or present twice, once in a left part of a base and once in a right part of another base,
\item the net of designs is acyclic and connex with respect to the graph of bases and cuts.
\end{itemize}
An address presents in a left part and in a right part defines a {\em cut}. 
\end{defi}

\begin{exa}\hfill
\begin{itemize}
\item A net of two designs of respective bases $\vdash \xi$ and $\xi \vdash$ forms a closed cut-net.
\item A net of two designs of respective bases $\xi \vdash \sigma$ and $\sigma \vdash \xi$ is not a closed cut-net: the graph of cuts is cyclic.
\item A net of three designs of respective bases $\vdash \xi$, $\xi \vdash \sigma$, $\sigma \vdash$ is a closed cut-net.
\end{itemize}
\end{exa}
 
\begin{defi}[Interaction on closed cut-nets]
Let $\design{R}$ be a closed cut-net. 
The design resulting from the interaction, denoted by $\normalisation{\design{R}}$, is defined in the following way:
let $\design{D}$ be the main design of $\design{R}$, with first action $\kappa$,
\begin{itemize}
\item if $\kappa$ is a daimon, then $\normalisation{\design{R}} = \{\daimon\}$,
\item otherwise $\kappa$ is a proper positive action $(+,\sigma,I)$ such that $\sigma$ is part of a cut with another design with last rule $(-,\sigma, {\mathcal N})$ (aggregating ramifications of actions on the same focus $\sigma$):
	\begin{itemize}
	\item If $I \not\in \mathcal N$, then interaction fails.
	\item Otherwise, interaction follows with the connected part of subdesigns obtained from $I$ with the rest of $\design{R}$.
	\end{itemize}
\end{itemize}
\end{defi}

\noindent Following this definition, either interaction fails, or it does not end, or it results in the design $\dai = \{\daimon\}$.
The definition of orthogonality follows:

\begin{defi}[Orthogonal, Behaviour]~
\begin{itemize}
\item Let $\design{D}$ be a design of base $\xi \vdash \sigma_1, \dots, \sigma_n$ (resp. $\vdash \sigma_1, \dots, \sigma_n$), \\
let $\design{R}$ be the net of designs $(\design{A}, \design{B}_1, \dots, \design{B}_n)$ (resp. $\design{R} = (\design{B}_1, \dots, \design{B}_n)$), where $\design{A}$ has base $\vdash \xi$ and $\design{B}_i$ has base $\sigma_i \vdash$, \\
$\design{R}$ belongs to $\design{D}^\perp$ if $\psdes{\design{D}}{\design{R}} = \dai$. 
\item Let $\designset{E}$ be a set of designs of the same base, $\designset{E}^\perp = \bigcap_{\design{D} \in \designset{E}} \design{D}^\perp$.
\item $\designset{E}$ is a {\em behaviour} if $\designset{E} = \designset{E}^{\perp\perp}$.
A behaviour is {\em positive} (resp. {\em negative}) if the base of its designs is positive (resp. negative).
\end{itemize}
\end{defi}

\begin{figure}
\begin{center}
\begin{tikzpicture}[x=70pt,y=30pt]
\node at (-.6,0) {$\design E=$};
\node at (0,0) {$(+,\xi,\{1,3\})$};
\node at (-1,1.3) {$(-,\xi.3,\{0\})$};
\node at (-1,2.6) {$(+,\xi.3.0,\emptyset)$};
\node at (0,1.9) {$(-,\xi.1,\{0\})$};
\node at (0,3.2) {$(+,\xi.1.0,\{0\})$};
\node at (1,2.5) {$(-,\xi.1,\{1\})$};
\node at (1,3.8) {$(+,\xi.1.1,\{0\})$};

\draw (0,.3) -- (-1,1);
\draw (0,.3) -- (0,1.6);
\draw (0,.3) -- (1,2.2);
\draw (-1,1.6) -- (-1,2.3);
\draw (0,2.2) -- (0,2.9);
\draw (1,2.8) -- (1,3.5);

\node at (3.6,0) {$=\design F$};
\node at (3,0) {$(-,\xi,\{1,3\})$};
\node at (3,1) {$(+,\xi.3,\{0\})$};
\node at (3,2) {$(-,\xi.3.0,\emptyset)$};
\node at (3,3) {$(+,\xi.1,\{1\})$};
\node at (2.4,4) {$(-,\xi.1.1,\{0\})$};
\node at (2.4,5) {$\daimon$};
\node at (3.6,4) {$(-,\xi.1.1,\{6,8\})$};
\node at (3.6,5) {$(+,\xi.1.1.6,\{1\})$};

\draw (3,.3) -- (3,.7);
\draw (3,1.3) -- (3,1.7);
\draw (3,2.3) -- (3,2.7);
\draw (3,3.3) -- (2.4,3.7);
\draw (2.4,4.3) -- (2.4,4.7);
\draw (3,3.3) -- (3.6,3.7);
\draw (3.6,4.3) -- (3.6,4.7);

\draw[dashed,->,red] (.45,0) -- (2.45,0);
\draw[dashed,->,red] (2.45,.1) -- (2.45,1);
\draw[dashed,->,red] (2.4,1) .. controls (.9,.9) .. (-.55,1.3);
\draw[dashed,->,red] (-.55,1.4) -- (-.55,2.6);
\draw[dashed,->,red] (-.5,2.6) .. controls (.1,2.5) and (1.7,1.9) .. (2.45,2);
\draw[dashed,->,red] (2.45,2.1) -- (2.45,3);
\draw[dashed,->,red] (2.4,3) .. controls (1.9,2.65) .. (1.45,2.5);
\draw[dashed,->,red] (1.45,2.55) -- (1.5,3.8);
\draw[dashed,->,red] (1.55,3.8) -- (1.9,4);
\draw[dashed,->,red] (1.9,4.1) -- (1.9,4.85);
\node[red] at (1.9,5) {$\daimon$};
\end{tikzpicture} 
\end{center}
\caption{Interaction (dashed line) between two (orthogonal) designs}
\label{fig:Interaction}
\end{figure}

\begin{exa}
We present below a few examples of orthogonals of simple designs.
\begin{itemize}
\item In Fig.~\ref{fig:Interaction}, we present the interaction between the design $\design E$ defined in the fourth item of example~\ref{design-as-dessin} and a design $\design F$ that is one of its orthogonals.
\item 
Designs that are orthogonal to the design $\design D$ defined in the third item of example~\ref{design-as-dessin} are the design $\dai_+$ based on $\vdash\xi$, and designs including the design $\design G$ given below:
$$
\begin{array}{cl}
\design G=\{&(+,\xi,\{1,2,3\})~;\\
 & (+,\xi,\{1,2,3\})(-,\xi.1,\{4,7\})\daimon\}
\end{array}
$$
As already mentioned, there is no way to distinguish the two distinct proof-like  presentations of the design $\design D$ by interaction with a design in its orthogonal.
\end{itemize}
\end{exa}

 
\section{Designs as sets of paths}\label{sec:Paths}

 We characterize in section~\ref{sec:Incarnation} incarnation of a set of designs $E$ as particular cliques of {\em visitable paths in $E$}. 
 A visitable path in a set of designs $E$ is a sequence of actions $\pathLL{p}$ in a design $\design{D}$ of $E$ which are visited during a normalization with a net of designs of $E^\perp$. A clique of such paths is a set of pairwise coherent set of paths, where the coherence relation defined in~\ref{defi:coherence_path} generalizes the one given for chronicles.
What is called in this paper a visitable path is a {\em play} in terms of Game Semantics. We should stress again the fact that designs are untyped objects: what is used by interaction in a design $\design{D}$ depends on the set of designs in which $\design{D}$ is considered. Concretely, it is possible to consider which designs are orthogonal to a given design $\design D$ alone, or to a set of designs $E$ that contains $\design D$. Sequences of actions of $\design D$ that are visited during interaction are not necessarily the same for the first case (orthogonals to the set containing $\design D$ alone) or for the second case (orthogonals to each design of $E$, hence also to $\design D$): the set (when closed under bi-orthogonality) in which $\design{D}$ is considered gives the type. 
In section~\ref{section:normalization_visitable}, we call {\em normalization path} of a design $\design{D}$ a sequence of actions which are visited by interaction with {\em some} net of designs, without taking care of the set in which we may consider the design $\design{D}$: a normalization path may be viewed as a potential play in terms of Game Semantics.

The remainder of this section is devoted to a study of {\em paths based on a net}. The {\em base of a net} consists in the sets of loci that serve as initial addresses, hence initial actions of a design. Paths generalize normalization paths (and visitable paths) in that their definition do not take into account the dual aspect of a play as we notice below after having introduced the concept of view. In studies concerning Linear Logic and done by means of Game Semantics (\eg\ \cite{DBLP:conf/csl/FaggianH02,Laurent05a}), a play in a strategy is a sequence of actions that satisfies several conditions among which {\em visibility}. Visibility is one of the constraints that ensure that a play is a denotation of a cut-elimination process.
Visibility is defined by means of a {\em view} operation on sequences (present in Hyland and Ong-Nickau games~\cite{HylandOng93,DBLP:conf/lfcs/Nickau94}).   
The view allows to recover a chronicle from a sequence of actions, hence also from paths: a view is the subsequence obtained from some sequence of actions by jumping from a negative action to its justification, if any, that should be a positive action, and from a positive action to the action that immediately precedes.
In Fig.~\ref{fig:ExNegJump} on the left, the view of the sequence $\pathLL{p}$ is the dashed red sequence noted $\view{\pathLL{p}}$. Note that the concept of view derived from the one of play in Game Semantics has a primitive r\^ole in Ludics: a design is a set of chronicles and chronicles are the views of potential plays, of a potential innocent strategy. In other words, the concept of view is derived from the one of play in Game Sematics whereas the concept of path is derived from the one of chronicle in Ludics.

\begin{defi}[Base, Sequence, View]~
\begin{itemize}
\item A {\bf base of net} $\beta$ is a non-empty finite set of sequents of pairwise disjoint loci: $\Gamma_1\vdash\Delta_1$, \dots,
$\Gamma_n\vdash\Delta_n$ such that each $\Gamma_i$ contains exactly one locus $\xi_i$, except at most one which may be empty, and the $\Delta_j$ are finite sets.
\item A sequence of actions $\pathLL{s}$ is {\bf based on $\beta$} if an action of $\pathLL{s}$ either is hereditarily justified by an element of one of the sets $\Gamma_i$ or $\Delta_i$, or is a daimon which is, in the later case, the last action of $\pathLL{s}$. 
An action is {\em initial} if it is an element of one of the sets $\Gamma_i$ or $\Delta_i$.
\item Let $\pathLL{s}$ be a sequence of actions based on $\beta$, the {\bf view} $\view{\pathLL{s}}$ is the subsequence of $\pathLL{s}$
defined as follows:
\begin{itemize}[label=$-$]
\item $\view{\epsilon}=\epsilon$ where $\epsilon$ is the empty sequence;
\item $\view{\kappa}=\kappa$;
\item $\view{w\kappa^+}=\view{w}\kappa^+$;
\item $\view{w\kappa^-}=\view{w_0}\kappa^-$ where $w_0$ either is empty if $\kappa^-$ is initial or is the prefix of $w$ ending with the positive action which justifies $\kappa^-$.
\end{itemize}
\end{itemize}
\end{defi}

\noindent Proponent and Opponent players, as coined in Game Semantics, arise naturally in Ludics from the fact that actions are polarized and polarity alternates in chronicles: with respect to a design $\design{D}$, a positive action is played by Proponent whereas a negative action is played by Opponent. Opponent and Proponent are interchanged when considering an orthogonal to the design $\design{D}$. 
A sequence of actions is {\em visible} if the justification of every Proponent action $\kappa$ is in the Proponent view of $\kappa$, this is called P-visitability, and dually, if the justification of every Opponent action $\kappa$ is in the Opponent view of $\kappa$ , called O-visitability. 
 A path, see Def.~\ref{defi:path}, is defined as an alternate, justified, linear, total sequence of actions that satisfies the P-visitability (Opponent being not already known).
We prove in section~\ref{section:normalization_visitable} that a normalization path is a path such that its dual is also a path: in other words the two facets of visibility are satisfied, \ie, P- and O-visitability.
The constraints (Totality) and (Daimon) are added to care of positive bases: in case the base contains a sequent $\vdash \Delta$, either the path is reduced to the daimon or it begins with a positive action focused on a locus of $\Delta$.
%

\input{FIGURES/fig_NegJump}

\begin{defi}[Path]\label{defi:path}
A {\bf path} $\pathLL{p}$ based on $\beta$ is a finite sequence of
actions based on $\beta$ such that
\begin{itemize}
\item {\em Alternation:} The polarity of actions alternate between
positive and negative.
\item {\em Justification:} A proper action is either justified, \ie, its
focus is built by one of the previous actions in the sequence, or it is
called initial with a focus in one of $\Gamma_i$ (resp. $\Delta_i$) if
the action is negative (resp. positive).
\item {\em Negative jump:} (There is no jump on positive action) Let $\pathLL{q}\kappa$ be a subsequence of $\pathLL{p}$,\\
- If  $\kappa$ is a positive proper action justified by a negative
action $\kappa'$ then $\kappa' \in \view{\pathLL{q}}$.\\
- If  $\kappa$ is an initial positive proper action then its focus belongs to
one $\Delta_i$  and either $\kappa$ is  the first action of $\pathLL{p}$
and $\Gamma_i$ is empty, or  $\kappa$ is immediately preceded in
$\pathLL{p}$ by a negative action with a focus hereditarily justified by an element
of $\Gamma_i\cup\Delta_i$.
\item {\em Linearity:} Actions have distinct focuses.
\item {\em Daimon:} If present, a daimon ends the path. If it is the
first action in $\pathLL{p}$ then one of $\Gamma_i$ is empty.
\item {\em Totality:} If there exists an empty $\Gamma_i$, then $\pathLL{p}$ is non-empty and begins either with $\daimon$ or with a positive action with a focus in $\Delta_i$.
\end{itemize}
\end{defi}

\noindent In the remainder of the paper, we use an explicit formulation for the justifier of an action to precede it in the sequence (negative jump): Let $\kappa$ be a positive proper action justified by a negative
action $\kappa'$, 
$\kappa' \in \view{\pathLL{q}}$
iff
there is a sequence
$\alpha_0^+\alpha^-_0\dots\alpha^-_n$  beginning with
$\kappa=\alpha_0^+$, ending with $\kappa'=\alpha_n^-$ and such that
$\alpha^-_i$ immediately precedes $\alpha^+_i$ in $\pathLL{p}$ and
$\alpha_{i+1}^+$ justifies $\alpha_i^-$.
Fig.~\ref{fig:ExNegJump} gives an example of a sequence that is not a path: the sequence $\pathLL{q}$ `jumps' from the negative action $(-,\xi.1.1.0,\{1\})$ in a chronicle to the positive action $(+,\xi.2.1,\{1\})$ that stays in another chronicle.
Remark also that the view $\view{\pathLL{q}}$ is not a chronicle. Proposition~\ref{chroniclebase} establishes that the view of a path is a chronicle, and we prove in proposition~\ref{prop:pathsTOnet} that views of a set of {\em pairwise coherent} paths define a net of designs.

\input{FIGURES/fig_exempleNegJump}

A path is {\em positive} or {\em negative} with respect to the polarity of its last action. If $P$ is a set of paths, we note $P^+$ the subset of $P$ consisting of positive paths, \ie, paths ending with a positive action.
Remark that a (non-empty in case the base is positive) prefix of a path is a path.
Remark also that a chronicle based on $\Gamma \vdash \Delta$ is a path.
Before we establish how paths may be viewed as chronicles, a few more facts must be stated.

 \begin{lem}\label{projection}
Let $\pathLL{p}$ be a path based on $\beta$, and $\kappa$ be a proper action occurring in $\view{\pathLL{p}}$. 
  \begin{itemize}
  \item If $\kappa$ is initial and negative in $\view{\pathLL{p}}$ then it is the only initial negative action occurring in $\view{\pathLL{p}}$ and it is the first action of $\view{\pathLL{p}}$.
  \item $\kappa$ is initial in $\pathLL{p}$  iff it is  initial in $\view{\pathLL{p}}$.
  \item $\kappa$ is justified by $\kappa'$ in $\pathLL{p}$ iff  $\kappa$ is justified by $\kappa'$ in $\view{\pathLL{p}}$.
  \end{itemize}
 \end{lem}

\begin{proof}
Let us observe that by construction if $\kappa$ is a proper action occurring in $\view{\pathLL{p}}$ and if $w\kappa$ is a prefix of $\pathLL{p}$ then $\view{w\kappa}$ is a prefix of $\view{\pathLL{p}}$.
\begin{itemize}
\item If  $\kappa$ is negative and initial in $\view{\pathLL{p}}$ then it is the first action of $\view{\pathLL{p}}$: we have $\pathLL{p} = w\kappa w'$ and $\view{\pathLL{p}} = \view{w\kappa}w'_0 = \kappa w'_0$. Hence it is unique.
\item By construction, if $\kappa$ is initial in $\pathLL{p}$ then it is initial in $\view{\pathLL{p}}$. 
Similarly, if $\kappa$ is justified by $\kappa'$ in $\view{\pathLL{p}}$ then  $\kappa$ is justified by $\kappa'$ in $\pathLL{p}$.
\end{itemize}
We prove the converse of the two last properties by considering separately the two polarities of $\kappa$.
\begin{itemize}
\item If $\kappa$ is negative and justified  by $\kappa'$ in $\pathLL{p}$ then, by construction, $\kappa$ is justified by $\kappa'$ in $\view{\pathLL{p}}$. Hence if $\kappa$ is initial in $\view{\pathLL{p}}$, it should be initial in $\pathLL{p}$.
\item If $\kappa$ is positive and justified by a negative action $\kappa'$ in $\pathLL{p}$ then there is a finite sequence 
 $\alpha_0^+\alpha^-_0,\dots\alpha^-_n$  beginning with $\kappa=\alpha_0^+$, ending with $\kappa'=\alpha_n^-$ and such that $\alpha^-_i$ immediately precedes $\alpha^+_i$ in $\pathLL{p}$ and $\alpha_{i+1}^+$ justifies $\alpha_i^-$. This means that the prefix $w\kappa$ of $\pathLL{p}$  may be written 
 $w_n\alpha^-_n\alpha^+_nw_{n-1}\dots\alpha^-_i\alpha^+_iw_i\dots\alpha_0^-\alpha^+_0$. Then 
\[
\view{w\kappa}=\view{w_n\alpha^-_n\alpha^+_nw_{n-1}\dots \alpha_0^-\alpha^+_0}
=\view{w_n}\kappa'\alpha^+_n\alpha^-_{n-1}\dots\alpha^-_i\alpha^+_i\dots\alpha_0^-\kappa
 \]
 Hence $\kappa$ is justified by $\kappa'$ in $\view{\pathLL{p}}$.
 Thus if $\kappa$ is initial in $\view{\pathLL{p}}$, then it should be initial in $\pathLL{p}$.
\end{itemize}
  \end{proof}

\begin{prop}\label{chroniclebase}
Let $\pathLL{p}$ be a non empty path based on $\Gamma_1\vdash\Delta_1$, \dots,  $\Gamma_n\vdash\Delta_n$. The sequence $\view{\pathLL{p}}$   is a chronicle based on $\Gamma_k\vdash\Delta_k$ for some $k\in\{1,\dots,n\}$.
\end{prop}

\begin{proof}
Since $\pathLL{p}$ contains at least one action, $\view{\pathLL{p}}$ contains at least $\kappa$ where  $\kappa$ is the last action of $\pathLL{p}$.
By construction $\view{\pathLL{p}}$ is alternate.  As a consequence of lemma~\ref{projection} the proper actions of $\view{\pathLL{p}}$ whose focuses are not in one of the bases (hence not initial in $\pathLL{p}$) are justified in $\view{\pathLL{p}}$. And of course if present the daimon ends $\view{\pathLL{p}}$. Moreover, a negative action which is not initial is justified by the positive action which immediately precedes it in $\view{\pathLL{p}}$. Hence $\view{\pathLL{p}}$ is a chronicle.
We just have still to check that the chronicle $\view{\pathLL{p}}$ is based on one $\Gamma_k\vdash \Delta_k$. This is done by induction on the length of $\pathLL{p}$.
\begin{itemize}
\item $\pathLL{p}=\kappa$:\\
 Suppose that $\kappa$ is negative. It is initial, its focus is in one $\Gamma_k$, then $\view{\pathLL{p}}$ is the chronicle $\kappa$ based on $\Gamma_k\vdash\Delta_k$. 
 \\
 Suppose that $\kappa$ is positive. In such a case, there exists an empty $\Gamma_k$ and either $\kappa=\daimon$ or the focus of $\kappa$ is in $\Delta_k$. Thus $\view{\pathLL{p}}=\kappa$ is a chronicle based on $\vdash\Delta_k$. 

  \item $\pathLL{p}=w\kappa_1\kappa$. By induction hypothesis, for each non-empty strict prefix $\pathLL{q}$ of $\pathLL{p}$, there exists $k_\pathLL{q}$ such that $\view{\pathLL{q}}$ is a chronicle based on a $\Gamma_{k_\pathLL{q}}\vdash\Delta_{k_\pathLL{q}}$.\\
   Suppose that $\kappa$ is negative. If it is initial, its focus is in one $\Gamma_k$,  then $\view{\pathLL{p}}$ is the chronicle $\kappa$ based on $\Gamma_k\vdash\Delta_k$.  Otherwise, by construction $\kappa$ is justified by the positive action $\kappa'$ which immediately precedes $\kappa$ in $\view{\pathLL{p}}$. Let $w_0\kappa'$ be the prefix of $\pathLL{p}$ ending by $\kappa'$. It is a strict non-empty prefix, hence there exists $k$ such that $\view{w_0\kappa'}$ is a chronicle based on $\Gamma_k\vdash\Delta_k$ and $\view{\pathLL{p}}=\view{w_0\kappa'}\kappa$ is also a chronicle based on $\Gamma_k\vdash\Delta_k$.\\
Suppose that $\kappa$ is positive. We have that $\view{\pathLL{p}}=\view{w\kappa_1}\kappa$ and $\view{w\kappa_1}$ is a chronicle based on some $\Gamma_k\vdash\Delta_k$. Note that $\view{w\kappa_1}\daimon$ is a chronicle based on $\Gamma_k\vdash\Delta_k$. So we suppose now that $\kappa \neq \daimon$. If $\kappa$ is initial then  its focus  is one $\delta_j\in\Delta_j$ and   the focus of $\kappa_1$ is hereditarily justified by an element of $\Gamma_j\cup\Delta_j$.  Hence $j=k$ and $\view{\pathLL{p}}$  is a chronicle based on   $\Gamma_k\vdash \Delta_k$. If $\kappa$ is justified by a negative action $\kappa'$  then $\pathLL{p}=w_0\kappa'w_1\kappa$ and, by lemma~\ref{projection}, $\kappa'$ is in $\view{\pathLL{p}}=\view{w_0\kappa'w_1}\kappa$. Hence, as by induction $\view{w_0\kappa'w_1}$ is a chronicle based on some $\Gamma_k\vdash\Delta_k$, so is $\view{\pathLL{p}}$.
  \end{itemize}
\end{proof}

\noindent In the following, we develop further on this idea and show in proposition~\ref{prop:pathsTOnet} that we can build in a unique way a net of designs from the views of a set of paths. For that purpose, we must consider the views of all {\em prefixes} of paths: each prefix may define a particular chronicle. Furthermore, not only such paths must obviously have the same base but they must also be pairwise {\em coherent}, where this coherence relation generalizes the notion of coherence given for chronicles.
 
\begin{defi}[Coherence on paths]\label{defi:coherence_path}
Two paths $\pathLL{p_1}$ and $\pathLL{p_2}$ on the same base are
coherent, noted $\pathLL{p_1} \coh \pathLL{p_2}$, when: 
\begin{itemize}[label=$-$]
\item their first actions have same polarity: either positive and the first actions are the same or negative;
\item for all sequences $w_1\kappa_1^+$ and  $w_2\kappa_2^+$ respectively prefixes of $\pathLL{p_1}$ and  $\pathLL{p_2}$:
\begin{itemize}[label=\quad]
\item if $\view{w_1}=\view{w_2}$ then $\kappa^+_1=\kappa^+_2$;
\end{itemize}
\item for all sequences $w_1\kappa_1^-$ and  $w_2\kappa_2^-$, respectively prefixes of $\pathLL{p_1}$ and  $\pathLL{p_2}$, \\
let $w^0_1$ be either the empty sequence if $\kappa_1^-$ is initial or the prefix of $\pathLL{p_1}$ ending by the justification of $\kappa_1^-$, \\
let $w^0_2$ be either the empty sequence if $\kappa_2^-$ is initial or the prefix of $\pathLL{p_2}$ ending by the justification of $\kappa_2^-$,
\begin{itemize}[label=\quad]
\item if $\view{w^0_1}=\view{w^0_2}$ and $\kappa_1^-$ and $\kappa_2^-$ have distinct focuses
\item then for all actions $\sigma_1$ and $\sigma_2$ such that $w_1\kappa_1^-w_1'\sigma_1$  and  $w_2\kappa_2^-w_2'\sigma_2$ are respectively prefixes of $\pathLL{p_1}$ and $\pathLL{p_2}$, and such that  $\kappa^-_1\in\view{w_1\kappa_1^-w_1'\sigma_1}$ and $\kappa_2^-\in\view{w_2\kappa_2^-w_2'\sigma_2}$, $ \sigma_1$ and $\sigma_2$ have distinct focuses.
\end{itemize}
\end{itemize}
To rephrase the last item, $\kappa_1^-$ and $\kappa_2^-$ satisfying the conditions, $\sigma_1$ and $\sigma_2$ should have distinct focuses if they are ulterior actions hereditarily justified by $\kappa_1^-$ and $\kappa_2^-$ respectively in $\pathLL{p_1}$ and $\pathLL{p_2}$.

\end{defi}

The relation $\coh$ is reflexive. 
Note that, applied to chronicles, the notion of coherence on paths coincide with the notion of coherence on chronicles given by Girard~\cite{DBLP:journals/mscs/Girard01} and recalled in the previous section.
Hence, in the following, we only use the notation $\coh$ to speak of coherence between paths or between chronicles.
Moreover, the comparability property in the definition of coherence between chronicles is still satisfied between coherent paths: if $\pathLL{p_1} \coh \pathLL{p_2}$ then either one extends the other or they first differ on negative actions. 
In fact, if $w\kappa_1^+$ and $w\kappa_2^+$ are respective prefixes of $\pathLL{p_1}$ and  $\pathLL{p_2}$ with $\kappa_1^+ \neq \kappa_2^+$, then the condition of coherence is not satisfied, hence $\pathLL{p_1} \not\coh \pathLL{p_2}$.\bigskip

We show below that the union of views of prefixes of coherent paths such that maximal ones are positive forms a net of designs. Taking a unique path, we recover the fact that it gives rise to a net of slices.\bigskip

\noindent{\sc Notations}:\hfill
\begin{itemize}[label=$-$]
\item Let $\pathLL{p}$ be a non-empty path, $\pathLL{p}^*$ is the prefix closure of $\pathLL{p}$, \ie, the set of paths that are non-empty prefixes of $\pathLL{p}$. We put $\epsilon^* = \{\epsilon\}$.
\item Let $D$ be a set of paths of the same base, $\view{D} = \{\view{\pathLL{p}} ; \pathLL{p} \in D\}$ and $D^* = \bigcup_{\pathLL{p} \in D} \pathLL{p}^*$.
\item Let $\pathLL{p}$ be a pathh, $\view{\pathLL{p}^*}$ is noted
  $\fullview{\pathLL{p}}$. More generally, $\view{D^*}$ is noted
  $\fullview{D}$.
\end{itemize}

\begin{prop}\label{prop:pathsTOnet}
 Let $D$ be a non empty set  of non empty coherent paths based   on   $\beta$
such that maximal ones are positive and let $D^*$ be its closure by prefixes. The set of chronicles $\fullviews{D}$ defined as the union of views of paths of $D^*$ forms a net of designs based on $\beta$. 
  \end{prop}

 \begin{proof}
Let $\beta = \Gamma_1\vdash \Delta_1, \dots, \Gamma_n\vdash \Delta_n$.
Let $\rho$ be either an element of one of $\Gamma_i$ or the focus of the first positive action of the paths of $D$ (and in this case there exists $i$ such that $\Gamma_i = \emptyset$). We consider the subset $D^*_\rho$ of $D^*$ defined in the following way: 
\[
D^*_\rho
=\{ \pathLL{p}\in D^*;\quad\view{\pathLL{p}}\mbox{ begins with an action focused on  }\rho\}
\] 
 We consider the set of chronicles defined by  $\views{D^*_\rho}=\{ \views{\pathLL{p}};\; \pathLL{p}\in D^*_\rho\}$.
Suppose $\views{D^*_\rho}$ is not empty. We know by proposition \ref{chroniclebase} that there exists $i$ such that these chronicles are all based on $\Gamma_i\vdash\Delta_i$, and either $\Gamma_i$ is empty and $\rho\in\Delta_i$ or $\Gamma_i = \{\rho\}$.
 Then  we check that $\views{D^*_\rho}$ is a design based on $\Gamma_i\vdash\Delta_i$:
\begin{itemize}
\item (Forest) Suppose that $\chronicle{c}\in \views{D^*_{\rho}}$ and let $w$ be a non empty prefix of $\chronicle{c}$. As $\chronicle{c}\in \views{D^*_{\rho}}$, there is a path $\pathLL{p}$ belonging to $D$ and a prefix  $\pathLL{p'}$ of $\pathLL{p}$, belonging to $D^*_\rho$, such that $\view{\pathLL{p'}}=\chronicle{c}$. The last action $\kappa$ of $w$ belongs to $\view{\pathLL{p'}}$. 
If $\kappa$ is the daimon, then it is the last action of $\chronicle{c}$, hence $w= \chronicle{c} \in \views{D^*_{\rho}}$.
Otherwise, $\kappa$ is a proper action occurring in $\view{\pathLL{p'}}$, hence in $\pathLL{p'}$, thus there exists $w_0$ and $w_0\kappa$ is a prefix of $\pathLL{p'}$. We already observed that in that case $\view{w_0\kappa}$ is a prefix of $\view{\pathLL{p'}}$ ending with $\kappa$. Hence $w=\view{w_0\kappa}$ belongs to $\views{D^*_{\rho}}$. 

\item (Coherence) Suppose that $w\kappa_1$ and $w\kappa_2$ belong to $\views{D^*_{\rho}}$. This means that there are two paths $\pathLL{p}_1$ and $\pathLL{p}_2$, belonging to $D$, and two prefixes $\pathLL{p}_1'$ and $\pathLL{p}_2'$ of respectively $\pathLL{p}_1$ and $\pathLL{p}_2$, belonging to $D^*_{\rho}$, such that $\view{\pathLL{p}_i'}=w\kappa_i$ for $i = 1,2$. 
Let us note $\pathLL{p}'_1=w_1\kappa_1$ and
$\pathLL{p}'_2=w_2\kappa_2$.
\begin{itemize}[label=$-$]
\item If $\kappa_1$ and $\kappa_2$ are positive then $\view{\pathLL{p}_i'}=\view{w_i\kappa_i} =\view{w_i}\kappa_i$ for $i = 1,2$, hence $\view{w_1}=w=\view{w_2}$. Thus, as $\pathLL{p}_1$ and  $\pathLL{p}_2$ are coherent, $\kappa_1 = \kappa_2$.
\item If $\kappa_1$ and $\kappa_2$ are negative with distinct focuses  then all the ulterior actions $\sigma_1$ and $\sigma_2$ hereditarely justified by $\kappa_1$ and $\kappa_2$ respectively in $\pathLL{p_1}$ and $\pathLL{p_2}$  have distinct focuses; then it is still true in $\view{\pathLL{p}_1} $ and $\view{\pathLL{p}_2}$, \ie, if $w_1\kappa_1^-w_1'\sigma_1$  and  $w_2\kappa_2^-w_2'\sigma_2$ are initial sequences respectively of $\pathLL{p_1}$ and $\pathLL{p_2}$ and are such that  $\kappa^-_1\in\view{w_1\kappa_1^-w_1'\sigma_1}$ and $\kappa_2^-\in\view{w_2\kappa_2^-w_2'\sigma_2}$, then
$ \sigma_1$ and $\sigma_2$ have distinct focuses. 
\end{itemize}

\item (Positivity) By construction since the maximal paths in $D$ are positive ones, chronicles in $\views{D^*_{\rho}}$ without extension end with a positive action.

\item (Totality) It follows from the Totality constraint on paths.
\end{itemize}
Note that sets $\views{D^*_{\rho}}$ are either empty or designs with disjoint sequents as bases. Hence the result.
\end{proof}

We know how to reconstruct some designs or nets of designs from cliques of paths. In fact a net of designs $\design{R}$ based on $\beta$ may be recovered from the set of {\em paths on $\design{R}$}, that is to say paths $\pathLL{p}$ based on $\beta$ such that $\fullview{\pathLL{p}}$ is a subnet of $\design{R}$:

\begin{prop}\label{prop:PathsNet}
 Let $\design{R}$ be a net of designs of base $\beta$. Let
 $\PoD{\design{R}}$ stand for the set of paths $\pathLL{p}$ based on
 $\beta$ such that $\fullview{\pathLL{p}}$ is a subnet of
 $\design{R}$. The following properties are satisfied:
\begin{itemize}[label=$-$]
\item If $\chronicle{c}$ is a chronicle of $\PoD{\design{R}}$ then it is a chronicle of $\design{R}$.
\item $\PoD{\design{R}}$ is prefix-closed.
\item If $\pathLL{p} \in \PoD{\design{R}}$ then $\view{\pathLL{p}} \in \PoD{\design{R}}$.
\item $\views{\PoD{\design{R}}} = \design{R}$
\end{itemize}
\end{prop}
\begin{proof}
As we noticed when defining a path, a prefix of a path is a path, a chronicle is a path, hence the first two items follow. Moreover a view of a path is a chronicle, hence a path.  
\end{proof}

We give in proposition~\ref{prop:characSetOfPaths} an inductive characterization of the set $\PoD{\design{R}}$. 
 We prove that chronicles (augmented first with a positive action if the base is positive) are paths, and that paths are obtained inductively by combining paths that are compatible in the sense given in proposition~\ref{createpath}:
  proposition~\ref{createpath} below gives a means for building a path $\pathLL{p}_1\pathLL{r}$ from two existing ones $\pathLL{p}_1$ and $\pathLL{p}_2 = \pathLL{q}\pathLL{r}$,
\ie\ to extend $\pathLL{p}_1$ by a suffix of $\pathLL{p}_2$: this extension is a {\em negative jump} at the end of $\pathLL{p}_1$, hence $\pathLL{p}_1$ must not end with a daimon, $\pathLL{r}$ must begin with a negative action, and actions in $\pathLL{r}$ must also satisfy a few compatibility conditions with respect to $\pathLL{p}_1$.
\begin{prop}\label{createpath}
Let $\pathLL{p}_1=w_1\kappa_1^-w'_1\kappa^+_1$ and $\pathLL{p}_2=w_2\kappa_2^-w'_2$ be two paths based on $\beta$.
The sequence $\pathLL{p}=w_1\kappa_1^-w'_1\kappa_1^+\kappa_2^-w'_2$ is
a path based on $\beta$ if the following conditions are satisfied:
\begin{itemize}[label=$-$]
\item $\kappa^+_1$ is a proper action,
\item $\view{w_2\kappa^-_2}w'_2$ is a path based on $\beta$,
\item actions in $\kappa_2^-w'_2$ and  $w_1\kappa_1^-w'_1\kappa^+_1$ have distinct focuses,
\item $\kappa^-_1$ and $\kappa_2^-$ are together initial, or are justified by the same positive action, or\footnote{In the case where $\beta$ contains a positive sequent $\vdash\Delta_i$ and $\kappa^+$ is anchored in $\vdash\Delta_i$.} $w_1=w_2=\kappa^+$, the action $\kappa_1^-$ is justified by $\kappa^+$ and $\kappa_2^-$ is initial, 
\item $\view{w_1\kappa^-_1}$ and $\view{w_2\kappa^-_2}$ coincide
  except on their last actions or are respectively equal to
  $\kappa^+\kappa_1^-$ and $\kappa^-_2$. 
\end{itemize}
\end{prop}

\begin{proof}
By construction  $\pathLL{p}$ is alternate, its proper actions have distinct focuses and are  either justified or   initial with a focus in $\beta$. Moreover, linearity and totality are satisfied and, if present, the daimon ends $\pathLL{p}_2$ hence ends the path. It remains to prove that the constraint (Negative Jump) is satisfied. Let $\kappa$ be an action in $\pathLL{p}$ then one has:\\
- If  $\kappa$ is a positive proper action occurring in $w'_2$ and  justified by a negative
action $\kappa'$ in $\pathLL{p}_2$, then $\kappa'$ belongs to $\view{w_2\kappa_2^-}w'_2$: the path $\view{w_2\kappa_2^-}w'_2$ should contain a justification for $\kappa$. Hence in $\view{w_2\kappa_2^-}w'_2$ there is a sequence
$\alpha_0^+\alpha^-_0\dots\alpha^-_n$  beginning with
$\kappa=\alpha_0^+$, ending with $\kappa'=\alpha_n^-$ and such that
$\alpha_{i+1}^+$ justifies $\alpha_i^-$ and
$\alpha^-_i$ immediately precedes $\alpha^+_i$ in $\view{w_2\kappa_2^-}w'_2$. Then  this sequence is completely included  into $\pathLL{p}$. \\
- If  $\kappa$ is a positive initial action occurring in $w'_2$ then its focus belongs to some $\Delta_i$ (in $\beta$) and  the only possibility is that  $\kappa$ is immediately preceded in
$\view{w_2\kappa_2^-}w'_2$ by a negative action with a focus hereditarily justified by an element
of $\Gamma_i\cup\Delta_i$.
 \end{proof}

\begin{prop}\label{prop:characSetOfPaths}
Let $\design{R}$ be a net of designs based on $\beta$.
The set of paths $\PoD{\design{R}}$ of $\design{R}$ may be exactly
characterized as follows:
\begin{itemize}[label=$-$]
\item if $\beta$ has only negative sequents, then $\design{R} \subset \PoD{\design{R}}$ and the empty path is in $\PoD{\design{R}}$;
\item if there exists a design $\design{D}$ of $\design{R}$ with a positive base, then $\design{D} \subset \PoD{\design{R}}$ and, for all chronicles $\chronicle{c}$ of $\design{R}$ beginning with a negative action, $\kappa^+\chronicle{c}\in\PoD{\design{R}}$ where $\kappa^+$ is the first action of chronicles of $\design{D}$;
\item if $\pathLL{p}_1=w_1\kappa_1^-w'_1\kappa^+_1$ and $\pathLL{p}_2=w_2\kappa_2^-w'_2$ are two paths belonging to $\PoD{\design{R}}$ and satisfying the conditions of proposition~\ref{createpath} then
$\pathLL{p}=w_1\kappa_1^-w'_1\kappa_1^+\kappa_2^-w'_2\in\PoD{\design{R}}$.
\end{itemize}
\end{prop}
\begin{proof}
Let us give a net $\design{R}$ and a set $\mathbb P_{\design{R}} =
\bigcup_{i\geq 0} \mathbb P^i_{\design{R}}$ defined inductively as:
$\forall i\geq 0, \mathbb P^i_{\design{R}} \subset \mathbb
P^{i+1}_{\design{R}}$ and
\begin{itemize}[label=$-$]
\item If $\beta$ has only negative sequents, then $\design{R} \subset \mathbb P^0_{\design{R}}$ and the empty path is in $\mathbb P^0_{\design{R}}$;
\item If there exists a design $\design{D}$ of $\design{R}$ with a positive base, 
let $\kappa^+$ be the first action of its chronicles 
then $\design{D} \subset \mathbb P^0_{\design{R}}$ and for all chronicle $\chronicle{c}$ of $\design{R}$ beginning with a negative action, $\kappa^+\chronicle{c}\in\mathbb P^0_{\design{R}}$ ;
\item If $\pathLL{p}_1=w_1\kappa_1^-w'_1\kappa^+_1$ and $\pathLL{p}_2=w_2\kappa_2^-w'_2$ are two paths belonging to $\mathbb P^i_{\design{R}}$ and satisfying the conditions of proposition~\ref{createpath} then
$\pathLL{p}=w_1\kappa_1^-w'_1\kappa_1^+\kappa_2^-w'_2\in\mathbb
  P^{i+1}_{\design{R}}$.
\end{itemize}

\noindent We show below that $\PoD{\design{R}} = \mathbb P_{\design{R}}$.
It is straightforward that $\mathbb P^0_{\design{R}} \subset
\PoD{\design{R}}$. Moreover, with notations as in the last item,
$\fullview{\pathLL{p}} \subset \fullview{\pathLL{p}_1} \cup
\fullview{\pathLL{p}_2}$. Hence by induction one proves that $\mathbb
P_{\design{R}} \subset \PoD{\design{R}}$.

One proves the other inclusion by induction on the length of a path
$\pathLL{p}$ of $\PoD{\design{R}}$:
\begin{itemize}[label=$-$]
 \item The property is immediate if $\pathLL{p}$ is empty, a single action or a chronicle.
 \item If $\pathLL{p} = \kappa^-\kappa^+$ then $\pathLL{p}$ is a
   chronicle, if $\pathLL{p} = \kappa^+\kappa^-$ then either
   $\pathLL{p}$ is a chronicle, or $\kappa^-$ is a chronicle and
   $\kappa^+$ is the first action of a design of $\design{R}$ with
   positive base: in all these three cases $\pathLL{p}$ belongs to
   $\mathbb P_{\design{R}}$.

\item Otherwise $\pathLL{p}$ has length at least $3$ and is not a chronicle, hence $\view{\pathLL{p}}$ has at least two actions less than $\pathLL{p}$. 
Note that the longest common suffix of $\pathLL{p}$ and $\view{\pathLL{p}}$ is not empty and begins with a negative action, otherwise it would be $\pathLL{p}$ hence $\pathLL{p}$ would be a chronicle. 
Let us note this common suffix $\kappa_2^-w'_2$ and set $\pathLL{p}=\pathLL{p}_1\kappa_2^-w'_2$. The path $\pathLL{p}_1$ is a non-empty path of $\PoD{\design{R}}$ hence by induction an element of some $\mathbb P^i_{\design{R}}$. Moreover $\pathLL{p}_1$ ends with a positive proper action.
Either $\pathLL{p}_1$ is reduced to a positive action, hence the base of $\design{R}$ is positive and the result follows from the second item in the definition of $\mathbb P_{\design{R}}$.
Or $\pathLL{p}_1$ contains at least two actions if $\beta$ is negative or three actions if $\beta$ is positive.
Let $\pathLL{p}_2$ be defined in the following way: if the base of $\design{R}$ is negative or if this base is positive and $\view{\pathLL{p}}$ has a first action positive then $\pathLL{p}_2 = \view{\pathLL{p}}$, otherwise let $\kappa^+$ be the first action of the design of positive base and set $\pathLL{p}_2 = \kappa^+\view{\pathLL{p}}$. In the two cases, $\pathLL{p}_2 \in \mathbb P^0_{\design{R}}$ and is a path.
Moreover, $\kappa_2^-w'_2$ is a suffix of $\view{\pathLL{p}}$. 
Paths $\pathLL{p}_1$ and $\pathLL{p}_2$ satisfy conditions of
proposition~\ref{createpath}:
\begin{itemize}[label=$\cdot$]
\item take $\kappa^+_1$ as the last action of $\pathLL{p}_1$, hence it is a proper action,
\item take $w_2$ such that $\pathLL{p}_2 = w_2\kappa^-_2w'_2$ then $\view{w_2\kappa^-_2}w'_2 = \view{\pathLL{p}_2}$ is a path based on $\beta$,
\item no actions respectively  in $\kappa_2^-w'_2$ and  $\pathLL{p}_1$ have same focus,
\item if $\kappa_2^-$ is initial then let $\kappa^-_1$ be the first (initial)  action of $\pathLL{p}_1$ if $\beta$ is negative or be the second action of $\pathLL{p}_1$ otherwise. If  $\kappa_2^-$ is justified by a positive action then this last one is followed in $\pathLL{p}_1$ by a negative action $\kappa^-_1$: if it were not the case this positive action would be the last of $\pathLL{p}_1$ hence in the common suffix of $\pathLL{p}$ and $\view{\pathLL{p}}$ yielding a contradiction, 
\item $\view{w_1\kappa^-_1}$ and $\view{w_2\kappa^-_2}$ coincide
  except on their last actions which are respectively $\kappa_1^-$ and
  $\kappa_2^-$.
\end{itemize}
\end{itemize}
Hence $\pathLL{p} \in \mathbb P^{i+1}_{\design{R}}$. 
This ends the induction on the length of paths of $\PoD{\design{R}}$.
Thus $\PoD{\design{R}} \subset \mathbb P_{\design{R}}$.
\end{proof}


\section{Normalization paths, visitable paths}\label{section:normalization_visitable}

In~\cite{DBLP:journals/mscs/Girard01}, Girard introduces the concept of {\em dispute}. In a few words, a dispute is a possible travel in a cut-net of designs during normalization. 
Here we will focus on such disputes and more precisely on sequences of actions followed up on one side of an interaction.
In the sequel of this paper, we restrict ourself to particular closed cut-nets. Recall that a closed cut-net is a net of designs where all addresses in bases are part of a cut. 
We distinguish in such a net one design with base $\beta$ of the form
 $\xi\vdash\sigma_1,\dots,\sigma_n$ (resp. equal to $\vdash\sigma_1,\dots,\sigma_n$) and a net of designs which base is noted $\beta^\perp$ and that is equal to $\vdash\xi$, $\sigma_1\vdash$, \dots, $\sigma_n\vdash$ (resp.   equal to $\sigma_1\vdash$, \dots, $\sigma_n\vdash$). 
The previous section was devoted to define paths in nets of designs.
This section is concerned with {\em normalization paths}, \ie, sequences of actions that may be followed in a normalization. 
We prove that a normalization path is a path. However not all paths are normalization paths as normalization is a kind of `mirror' process: a path $\pathLL{p}$ may be followed by normalization iff its {\em dual} noted $\dual{\pathLL{p}}$ is a path, where $\dual{\pathLL{p}}$ has actions of opposite polarities and terminating with a daimon exactly  when $\pathLL {p}$ does not end with a daimon. 
Then we determine when $\dual{\pathLL{p}}$ is a path as a condition on $\pathLL{p}$.
We finish this section with a more general question: let $\designset{E}$ be a set of designs of the same base, is it possible to characterize paths on $\design{D} \in \designset{E}$ that may be visited by an element of $\designset{E}^\perp$?

\begin{defi}[Normalization path]
Let $(\design{D},\design{R}$) be a convergent closed cut-net such that all the cut loci belong to the base of $\design{D}$. The {\em normalization path} of the interaction of $\design{D}$ with $\design{R}$, denoted $\normalisationSeq{\design{D}}{\design{R}}$, is the sequence of actions of $\design{D}$ visited during the normalization. It  may be defined by induction on the number $n$ of normalization steps:
\begin{itemize}
\item Case $n=1$:
\begin{itemize}
\item If the interaction stops in one step: either $\design{D}=\dai$,
  in this case $\normalisationSeq{\design{D}}{\design{R}}=\daimon$, or
  the main design (which is not $\design{D}$) is equal to $\dai$ and
  in this case $\normalisationSeq{\design{D}}{\design{R}}$ is the
  empty sequence.
\item Otherwise let $\kappa^+$ be the first action of the main
  design. The first action of
  $\normalisationSeq{\design{D}}{\design{R}}$ is $\kappa^+$ if
  $\design{D}$ is the main design and is $\overline{\kappa^+}$
  otherwise\footnote{Where the notation $\overline{\kappa}$ is simply
    $\Overline{(\pm,\xi,I)} = (\mp,\xi,I)$ and may be extended on
    sequences by $\Overline{\epsilon} = \epsilon$ and
    $\Overline{w\kappa} = \Overline{w}\,\Overline{\kappa}$.}.
\end{itemize}

\item Case $n=p+1$: the prefix $\kappa_1\dots\kappa_p$ of $\normalisationSeq{\design{D}}{\design{R}}$ is  already defined.\\
Either the interaction stops and $\normalisationSeq{\design{D}}{\design{R}}=\kappa_1\dots\kappa_p$ if the main design is a subdesign of $\design{R}$, or $\normalisationSeq{\design{D}}{\design{R}}=\kappa_1\dots\kappa_p\daimon$ if the main design is a subdesign of $\design{D}$.\\
Or, let $\kappa^+$ be the first proper action of the closed cut-net obtained after step $p$, $\normalisationSeq{\design{D}}{\design{R}}$ begins with $\kappa_1\dots\kappa_p\overline{\kappa^+}$ if the main design is a subdesign of $\design{R}$, or it begins with $\kappa_1\dots\kappa_p\kappa^+$ if the main design is a subdesign of $\design{D}$.
\end{itemize}
We note $\normalisationSeq{\design{R}}{\design{D}}$ the sequence of actions visited in $\design{R}$ during the normalization with $\design{D}$.
\end{defi}

The previous definition makes sense because normalization of convergent closed cut-nets is a step by step procedure that satisfies lemma~\ref{lemme:cut-nets}.

\begin{lem}\label{lemme:cut-nets}
Let $(\design{D},\design{R}$) be a convergent closed cut-net such that the base of $\design{D}$ contains all the cut-loci.
Every normalization step (until the last one) is a convergent closed
cut-net  $(\design{D}',\design{R}'$) where:
\begin{itemize}[label=$-$]
\item $\design{D}'$ is a subdesign of $\design{D}$,
\item either $\design{D}'$ is the main design, or $\design{D}'$ is
  such that its first action is the dual of the one of the main
  design, that is a subdesign of $\design{R}$, or $\dai$ is a design
  of $\design{R}'$.
\end{itemize}
\end{lem}

\proof By induction on the number of normalization steps, we prove that a) it occurs in a convergent closed cut-net  $(\design{D}',\design{R}'$) where $\design{D}'$ is a subdesign of $\design{D}$ and is either the main design or such that its first action is the dual of the one of the main design (except that if the main design is $\dai$) and b)  every newly  created  cut-locus appears in the base of a subdesign of $\design{D}$ and in the base of a subdesign of $\design{R}$.
\begin{enumerate}
\item We consider the first normalization step. The property  b) is satisfied by hypothesis; the property a) is deduced from it:
\begin{itemize}
\item either ${\design{D}}$ is the main design, it is located on a base $\vdash\sigma_1,\dots,\sigma_n$ while the other designs of the closed cut-net are located on $\sigma_i\vdash$,
\item or $\design{D}$ is located on $\xi\vdash\sigma_1,\dots,\sigma_n$. Since the base of $\design{D}$ contains all the cut loci, the other designs are  $\design{A}$ based on  $\vdash \xi$, and the designs $\design{B}_i$ based on $\sigma_i \vdash$. The main design is $\design{A}$ and since the interaction does not diverge, either $\design{A}$ is  $\dai$ or its first positive action $\kappa^+$ is such that the dual action $\overline{\kappa^+}$ is one of the first actions of $\design{D}$, then a) is satisfied.
\end{itemize}
\item We prove that the properties a) and b) are preserved during a normalization step:
\begin{itemize}
\item We consider a convergent closed cut-net $(\design{D}',\design{R}')$ where $\design{D}'$ is a subdesign of $\design{D}$ and is the main design based on $\vdash \sigma, \Delta$. If $\design{D}'$ is not $\dai$, let $\sigma$ be the focus of its first action. The next normalization step is on the closed cut-net which contains the same designs as before except that $\design{D}'$ is replaced by several designs $\design{D}'_i$ based on $\sigma i\vdash\Delta_i$ and the design based on $\sigma\vdash\Gamma$ is replaced by its subdesign $\design{A}'$ based on $\vdash \sigma 1,\dots,\sigma n,\Gamma$. The main design is now $\design{A}'$:
\begin{itemize}
\item Either its first action is $\daimon$  hence a) is satisfied, 
\item or it is a proper action $(+,\xi,I)$, where $\xi=\sigma {i_0}$ or $\xi\in\Gamma$. In the first case, since the interaction does not diverge, the dual action $(-,\sigma {i_0},I)$ is one of the first actions of the subdesign $\design{D}'_{i_0}$ of $\design{D}$, then a) is satisfied. Otherwise $\xi\in\Gamma$. The loci belonging to $\Gamma$ are not initial as in such a case they should appear in the base of $\design{D}$ in negative position, and this is not possible. So $\xi$ has been created during the normalization, and then is located in   subdesigns of $\design{D}$, that is that there is a subdesign of $\design{D}$ based on $\xi\vdash\Xi$ and (since the interaction does not diverge) having $(-,\xi,I)$ as first action. Then a) is satisfied. 
\end{itemize}
Moreover the only new cut loci are the $\sigma {i}$'s which are located on subdesigns of $\design{D}$. Hence b) is satisfied.
\item We consider a convergent closed cut-net in which the main design $\design{A}'$ is not $\dai$ and is such that the dual of its first action $(+,\sigma,I)$ is one of the first negative action of a subdesign $\design{D}'$ based on $\sigma\vdash\Gamma$.  The next normalization step is on the closed cut-net which contains the same designs as before except that $\design{D}'$ is replaced by its subdesign $\design{D}''$ based on $\vdash\sigma 1,\dots,\sigma n,\Gamma$ and some subdesigns $\design{A}_i'$ based on $\sigma i\vdash\Delta_i$. Since $\design{D}''$ is the main design, a) is satisfied. Moreover the only new cut loci are the $\sigma {i}$'s which are located on $\design{D}''$ which is a subdesign of $\design{D}$. Hence b) is satisfied.
\end{itemize}
\end{enumerate}

The following proposition shows that normalization paths are paths.

\begin{prop}\label{prop:NormPath}
 Let $(\design{D},\design{R})$ be a convergent closed cut-net.\\
$\normalisationSeq{\design{D}}{\design{R}}$ is a path on $\design{D}$.
$\normalisationSeq{\design{R}}{\design{D}}$ is a path on $\design{R}$.
\end{prop}
\begin{proof}
We prove the first property by induction on the length of a prefix of $\normalisationSeq{\design{D}}{\design{R}}$.\\
- The base case of the induction depends on the polarity of $\design{D}$: If $\design{D}$ has a negative base then the empty sequence is a path on $\design{D}$, otherwise $\design{D}$ has a positive base hence there exists a first action in the sequence $\normalisationSeq{\design{D}}{\design{R}}$, this action being the first action of $\design{D}$, hence a path on $\design{D}$.\\
- Suppose $\kappa_1\dots\kappa_p\kappa$ is a prefix of $\normalisationSeq{\design{D}}{\design{R}}$ and that by induction hypothesis $\kappa_1\dots\kappa_p$ is a path on $\design{D}$.\\
If $\kappa$ is a positive action then, with respect to normalization, $\view{\kappa_1\dots\kappa_p}\kappa$ is a chronicle of $\design{D}$ that extends $\view{\kappa_1\dots\kappa_p}$, hence $\kappa_1\dots\kappa_p\kappa$ is a path on $\design{D}$.\\
If $\kappa$ is an initial negative action hence $\design{D}$ is negative and $\kappa$ is the first action of the normalization, \ie, $p=0$, and $\kappa$ is a path on $\design{D}$.\\
Otherwise the focus of the negative action $\kappa$ has been created during normalization by a positive action present in $\kappa_1\dots\kappa_p$, hence $\kappa_1\dots\kappa_p\kappa$ is a path on $\design{D}$.\\
The second property is proved in the same way.
\end{proof}

Note that $\normalisationSeq{\design{R}}{\design{D}}$ is obtained from $\normalisationSeq{\design{D}}{\design{R}}$ by just changing polarities of proper actions and adding a daimon if it were not present in $\normalisationSeq{\design{D}}{\design{R}}$. 
More precisely we define the {\em dual} of a positive alternate sequence of actions $\pathLL{p}$ (possibly empty) to be a positive alternate sequence of actions $\dual{\pathLL{p}}$ (possibly empty) in the following way:
\begin{itemize}
\item If $\pathLL{p} = w\daimon$, $\dual{\pathLL{p}} := \Overline{w}$.
\item Otherwise $\dual{\pathLL{p}} := \Overline{\pathLL{p}}\daimon$.
\end{itemize}
 Note that $\dual{\dual{\pathLL{p}}} = \pathLL{p}$.
It follows from the definition of a normalization path that $\design{D} \perp \design{R}$ iff $\exists \pathLL{p} \in \PoD{\design{D}}$ such that $\dual{\pathLL{p}} \in \PoD{\design{R}}$.
Such a $\pathLL{p}$ is unique and is in fact $\normalisationSeq{\design{D}}{\design{R}}$.
Not all paths of $\PoD{\design{D}}$ may be visited by a net in $\design{D}^\perp$ as we show in the next example proposed by Faggian~\cite{DBLP:journals/tcs/Faggian06}.
\begin{exa}
Consider the design $\design{D}$ below:
\begin{center}
$
\design{D} = 
\scalebox{1}{
\infer{\vdash \xi,\sigma,\tau}{
	\infer{\xi0 \vdash \sigma}{
		\infer{\vdash \xi00,\sigma}{\sigma0 \vdash \xi00}
	}
	&
	\infer{\xi1 \vdash \tau}{
		\infer{\vdash \xi10,\tau}{\tau0 \vdash \xi10}
	}
}
}
$
\end{center}
Its orthogonal $\design{D}^\perp$ is the following set of nets (where dots may be replaced by any forest of actions such that the result is a design):
$$\raisebox{3ex}{\mbox{$\left\{\vphantom{\mbox{\infer{\xi \vdash}{\infer[\daimon]{\vdash \xi0, \xi1}{}}}}\right.$}}
\infer{\xi \vdash}
	{
	\infer[\daimon]{\vdash \xi0, \xi1}{}
	}
~
\deduce{\sigma\vdash}{\vdots}
~
\deduce{\tau\vdash}{\vdots}
\quad;\quad
\infer{\xi \vdash}
	{
	\infer{\vdash \xi0, \xi1}
		{
		\deduce{\xi00 \vdash \xi1}{\vdots}
		}
	}
~
\infer{\sigma\vdash}
	{
	\infer[\daimon]{\vdash \sigma0}{}
	}
~
\deduce{\tau\vdash}{\vdots}
\quad;\quad
\infer{\xi \vdash}
	{
	\infer{\vdash \xi0, \xi1}
		{
		\deduce{\xi10 \vdash \xi0}{\vdots}
		}
	}
~
\deduce{\sigma\vdash}{\vdots}
~
\infer{\tau\vdash}
	{
	\infer[\daimon]{\vdash \tau0}{}
	}
\raisebox{3ex}{\mbox{$\left.\vphantom{\mbox{\infer{\xi \vdash}{\infer[\daimon]{\vdash \xi0, \xi1}{}}}}\right\}$}}$$\medskip

\noindent Let us consider the path $\pathLL{p} = {(+,\xi,\{0,1\})(-,\xi0,\{0\})(+,\sigma,\{0\})(-,\xi1,\{0\})(+,\tau,\{0\})}$. It is a path of $\PoD{\design{D}}$, however it is not visited by a net in $\design{D}^\perp$: its dual $\dual{\pathLL{p}}$ is $(-,\xi,\{0,1\})(+,\xi0,\{0\})$ $(-,\sigma,\{0\})(+,\xi1,\{0\})(-,\tau,\{0\})\daimon$ that is not a path.
\end{exa}
In~\cite{DBLP:journals/tcs/Faggian06}, Faggian precisely characterizes {\em strong slices},  \ie\  slices of a design $\design{D}$ that may be completely visited with a unique normalization by an orthogonal of $\design{D}$. 
Remember that a slice of a design is a {\em multiplicative} subdesign of $\design{D}$, \ie\ a focus appears only in a unique negative action.
Furthermore, Faggian showed that a slice is strong iff it is finite and a partial preorder traversal may be defined on it: the root is visited, then in preorder each of its subtrees, the last visited node for each visited subtree being a leaf.
We remark that if $\pathLL{p}$ is a normalization path of a design $\design{D}$ then $\fullview{\pathLL{p}}$ is a strong slice of a $\design{D}$: the normalization process defines such a partial order on the design $\fullview{\pathLL{p}}$, and by construction $\fullview{\pathLL{p}} \subset \design{D}$ such that it does not contain two negative actions with the same focus, \ie\ $\fullview{\pathLL{p}}$ is a slice of $\design{D}$.
To summarize, let $\pathLL{p}$ be a sequence of actions based on $\beta$, we have that $\fullview{\pathLL{p}}$ is a strong slice iff $\pathLL{p}$ and $\dual{\pathLL{p}}$ are paths.

Note that the fact that $\fullview{\pathLL{p}}$ is a strong slice may be determined by means of constraints  only on $\pathLL{p}$. In fact, a dual of a path $\pathLL{p}$ is a path if $\pathLL{p}$ satisfies the dual of the first item of the Negative Jump constraint we call Restrictive Negative Jump. 
We notice that the Negative Jump together with Restrictive Negative Jump constraints are equivalent to the Visibility constraint in Game Semantics.
\begin{prop}\label{prop:dualpath_restNegJump}
 Let $\pathLL{p}$ be a path based on $\beta$, $\dual{\pathLL{p}}$ is a path based on $\beta^\perp$ iff $\pathLL{p}$ satisfies the following condition:
\begin{itemize}
\item {\em Restrictive Negative Jump:} for all sequence $\pathLL{q}\kappa$ that is a subsequence of $\pathLL{p}$, if  $\kappa$ is a negative action justified by a positive action $\kappa'$ then $\Overline{\kappa'} \in \view{\,\dual{\pathLL{q}}\,}$.
\end{itemize}
\end{prop}

\noindent The Restrictive Negative Jump constraint may be explicitly stated as follows:\\
{\em If  $\kappa^-$ is a proper action justified by an
action $\kappa'^{+}$ then there is a sequence
$\alpha_0^-\alpha^{+}_0\dots\alpha^{+}_n$  beginning with
$\kappa^-=\alpha_0^-$, ending with $\kappa'^{+}=\alpha_n^{+}$ and such that
$\alpha^{+}_i$ immediately precedes $\alpha^-_i$ in $\pathLL{p}$ and
$\alpha_{i+1}^-$ justifies $\alpha_i^{+}$.}

\begin{proof}
Suppose $\dual{\pathLL{p}}$ is a path then it satisfies the (Negative
Jump) constraint hence $\pathLL{p}$ satisfies its dual, hence in
particular the (Restrictive Negative Jump) constraint.

Suppose $\pathLL{p}$ satisfies (Restrictive Negative Jump). Being given the definition of $\dual{\pathLL{p}}$ and the fact that $\pathLL{p}$ is a path, the sequence $\dual{\pathLL{p}}$ satisfies conditions (Alternation), (Justification), (Linearity), (Daimon). It remains to check that $\dual{\pathLL{p}}$ satisfies (Totality) and (Negative Jump):
\begin{itemize}
\item (Totality) Note that if $\beta$ is negative then either $\pathLL{p}$ is empty then $\dual{\pathLL{p}} = \daimon$ is non-empty, or the first action of $\pathLL{p}$ is a negative action hence $\dual{\pathLL{p}}$ is non-empty.
\item (Negative Jump)
\begin{itemize}
\item The first item is satisfied by $\dual{\pathLL{p}}$ as $\pathLL{p}$ satisfies the Restrictive Negative Jump constraint.
\item Concerning the second item, let $\kappa$ be an initial positive proper action of $\dual{\pathLL{p}}$, then $\Overline{\kappa}$ is negative initial in $\pathLL{p}$, hence the first action of $\pathLL{p}$ being given the kind of base we consider. The result follows.\qedhere
\end{itemize}
\end{itemize}
\end{proof}

\noindent This achieves the characterization of what can be visited in an interaction: let $\pathLL{p}$ be a path in a design $\design{D}$, $\pathLL{p}$ may be visited in an interaction if, by definition, there exists a design $\design{E}$ and $\pathLL{p} = \normalisationSeq{\design{D}}{\design{E}}$. If $\dual{\pathLL{p}}$ is a path then $\fullview{\dual{\pathLL{p}}}$ is a design and we have that $\pathLL{p} = \normalisationSeq{\design{D}}{\fullview{\dual{\pathLL{p}}}}$. Finally proposition~\ref{prop:dualpath_restNegJump} gives the constraints that $\pathLL{p}$ should satisfy for $\dual{\pathLL{p}}$ to be a path.

We now move on to a more general question: let $\designset{E}$ be a set of designs of the same base, is it possible to characterize paths $\pathLL{p}$ on designs of $\designset{E}$ that may be visited by an orthogonal, \ie, an element of $\designset{E}^\perp$? A necessary condition is that $\dual{\pathLL{p}}$ is also a path, \ie, $\fullview{p}$ is a strong slice. Note that there is no reason that $\fullview{\dual{\pathLL{p}}}$ may be orthogonal to {\em all} the elements of $\designset{E}$. Indeed, the fact that $\dual{\pathLL{p}}$ is a path is not a sufficient condition: suppose a design $\design{D} \in \designset{E}$ with a first negative action $\kappa^-$, then $\pathLL{p} = \kappa^-$ is a path on $\design{D}$; $\pathLL{p}$ is visited by interaction with an element $\design{R}$ of $\designset{E}^\perp$ if $\design{R}$ contains a design with first positive action $\Overline{\kappa^-}$; but as $\design{R} \in \designset{E}^\perp$, $\kappa^-$ should be one of the first negative actions of each design of $\designset{E}$. 
This is the keypoint we use in proposition~\ref{stabilite-neg} for 
having a necessary condition for a path to be
 {\em visitable} in a set of designs $\designset{E}$. Note that this 
condition
 does not make direct reference to elements of $\designset{E}^\perp$.

\begin{defi}[Visitability in a Set of Designs]
Let $\designset{E}$ be a set of designs of the same base,
a path $\pathLL{p}$ is {\em visitable in $\designset{E}$} if there exists a design $\design{D}\in \designset{E}$ and a net $\design{R} \in \designset{E}^\perp$ such that $\pathLL{p} =  \normalisationSeq{\design{D}}{\design{R}}$. 
\end{defi}

\rem A visitable path $\pathLL{p}$ is positive (\ie, its last action is a positive one). Obviously, with notations as in the definition, $\pathLL{p}\in \PoD{\design{D}}$; moreover $\forall \design{D}'\in \designset{E}$ such that $\pathLL{p}\in \PoD{\design{D}'}$ then $\pathLL{p} =  \normalisationSeq{\design{D}'}{\design{R}}$ as normalization is deterministic.

\begin{prop}\label{stabilite-neg}
Let $\designset{E}$ be a set of designs and let $\pathLL{p}$ be a positive path of a design of $\designset{E}$,
\begin{center}
If the path $\pathLL{p}$ is {\em visitable in $\designset{E}$} 
then  
$\begin{cases}
\bullet~ \text{the sequence } \dual{\pathLL{p}} \text{ is a path,}\\
\bullet~ \text{for all prefix } w\kappa \text{ of } \pathLL{p}\text{, for all design } \design{D} \text{ in } \designset{E} \text{ such that }w \text{ is a path of } \design{D}\text{,}\\
\quad\text{if } \kappa \text{ is a negative action then } w\kappa \text{ is a path of } \design{D},\\
\end{cases}$
\end{center}
\end{prop}
\begin{proof}
Suppose that $\pathLL{p}$ is a path visitable in $\designset{E}$, then there exists a design $\design{D}_0\in \designset{E}$ and a net $\design{R}_0 \in \designset{E}^\perp$ such that $\pathLL{p} =  \normalisationSeq{\design{D}_0}{\design{R}_0}$, hence $\dual{\pathLL{p}} = \normalisationSeq{\design{R}_0}{\design{D}_0}$ is a path (in $\design{R}_0$). Let $w\kappa$ be a prefix of $\pathLL{p}$ and $\design{D}$ be a design of $\designset{E}$ such that $w\in\PoD{\design{D}}$. 
Note that $w$ cannot end with a daimon.
Since $w$ is a path of $\design{D}$, and that normalization is deterministic and $\design{D}\perp\design{R}_0$, then $w$ is a prefix of $\normalisationSeq{\design{D}}{\design{R}_0}$.
Suppose first that $\kappa$ is a negative action:
\begin{itemize}
\item If $w$ is empty, the base of designs of $\designset{E}$ is negative, \ie, of the form $\xi \vdash \sigma_1, \dots, \sigma_n$ and nets $\design{R}$ in $\designset{E}^\perp$ have bases $\vdash \xi$, $\sigma_1 \vdash$, ..., $\sigma_n \vdash$. So normalizations between $\design{R}$ and designs in $\designset{E}$ must begin with the same first positive action $\overline{\kappa}$ of $\design{R}$, hence $\kappa$ should be a first action in all designs in $\designset{E}$, thus a path in all designs in $\designset{E}$.
\item Otherwise, we note that after $|w|$ steps of normalization between $\design{D}$ and $\design{R}_0$, the cut-net consists of a net $\design{X}_\design{D}$ of subdesigns of $\design{D}$ and a net $\design{X}_{\design{R}_0}$ of subdesigns of designs of $\design{R}_0$. Moreover, there is exactly one design of positive base in this cut-net and this design belongs to $\design{X}_{\design{R}_0}$ as $w$ ends with a positive action.
Furthermore, $\design{X}_{\design{R}_0}$ is also the net obtained from $\design{R}_0$ after $|w|$ steps in the normalization between $\design{D}_0$ and $\design{R}_0$. Then $\overline{\kappa}$ is the next action to be used by normalization with 
$\design{D}_0$, hence also with $\design{D}$ (for normalization to proceed). Thus $w\kappa$ is a path of  $\design{D}$.\qedhere
\end{itemize}
\end{proof}

\begin{defi}[Completion of designs]\label{defi:closure}
Let $\design{D}$ be a design of base $\beta$, the {\tt completion} of $\design{D}$, noted $\design{D}^c$, is the design of base $\beta$ obtained from $\design{D}$ in the following way:\\
\centerline{
	$\design{D}^c := \design D \cup \{\chronicle c \kappa^-\daimon ~;~ \chronicle c \in \design D, \chronicle c \kappa^-\not\in \design D, \chronicle c\kappa^-\daimon$ is a chronicle of base $\beta\}$
}\\
Let $\design R$ be a net of designs, the {\em completion} of $\design R$ also written $\design R^c$ is the net of completions of designs of $\design R$.
\end{defi}

\noindent Note that $\design{D}^c$ is a design:
First, an action $\kappa^-$ is either initial or justified by the last action of $\chronicle c$ (in $\chronicle c \kappa^-$) hence linearity is satisfied.
Second, as chronicles $\chronicle c$ are in $\design D$ and $\kappa^-$ are negative actions, chronicles $\chronicle c \kappa^-\daimon$ are pairwise coherent and coherent with actions of $\design D$.

\begin{prop}\label{prop:completed_subdesign}
Let $\design{D}$ be a design in a behaviour $\behaviour{A}$, let $\design{C} \subset \design{D}$ then $\design{C}^c \in \behaviour{A}$.
\end{prop}
\begin{proof}
Let $\design{E} \in \behaviour{A}^\perp$. Hence $\design{E} \perp \design{D}$. Let $\pathLL{p}$ be the longest positive-ended path in the design $\design{C}$ that is a prefix of $\normalisationSeq{\design{D}}{\design{E}}$. 
Either $\pathLL{p} = \normalisationSeq{\design{D}}{\design{E}}$, hence $\design{E} \perp \design{C}$, and also $\design{E} \perp \design{C}^c$. 
Or there exist actions $\kappa^-$, $\kappa^+$ and a sequence $w$ such that $\normalisationSeq{\design{D}}{\design{E}} = \pathLL{p}\kappa^-\kappa^+w$. 
Consider the chronicle $\chronicle{c}$ such that $\view{\pathLL{p}\kappa^-} = \chronicle{c}\kappa^-$. By construction, $\chronicle{c} \in \design{C}$. Either $\chronicle{c}\kappa^- \in \design{C}$ hence also $\chronicle{c}\kappa^-\kappa^+ \in \design{C}$ as $\design{C} \subset \design{D}$ and there is a unique positive action after a negative action. Contradiction as $\pathLL{p}$ is then not maximal. 
Or $\chronicle{c}\kappa^-\daimon \in \design{C}^c$ hence $\design{E} \perp \design{C}^c$.
\end{proof}

The proposition~\ref{prop:completed_subdesign} is also true when we have nets of designs instead of designs.

\begin{prop}\label{prop:visitable-carac}
Let $E$ be a set of designs of same base.
Let $\pathLL p$ be a path of a design of $E$.
$\pathLL p$ is visitable in $E$ iff $\fullview{\dual{\pathLL p}}^c \in E^{\perp}$.
\end{prop}
\begin{proof}
Suppose that $\pathLL p$ is visitable in $E$. Then there exist $\design D \in E$ and $\design R \in E^\perp$ such that $\pathLL p =  \normalisationSeq{\design{D}}{\design{R}}$. 
Furthermore $\dual{\pathLL p}$ is a path in $\design R$, hence $\fullview{\dual{\pathLL p}} \subset \design R$. It follows from proposition~\ref{prop:completed_subdesign} that $\fullview{\dual{\pathLL p}}^c \in E^\perp$. \\
Suppose that $\fullview{\dual{\pathLL p}}^c \in E^{\perp}$, let $\design D$ be the design in $E$ such that $\pathLL p$ is a path of $\design D$. Note that $\design D \perp \fullview{\dual{\pathLL p}}^c$ and that $\pathLL p = \normalisationSeq{\design{D}}{\fullview{\dual{\pathLL p}}^c}$.
\end{proof}

Proposition~\ref{prop:visitable-carac} gives a means to compute the set of visitable paths of a set $E$ of designs of the same base (when $E$ is a finite set of finite designs): take each positive-ended path $\pathLL{p}$ of some design of $E$, test if for all design $\design D$ in $E$ we have $\design D \perp \fullview{\dual{\pathLL p}}^c$. This method may be improved by considering only paths that satisfy the necessary constraint given in proposition~\ref{stabilite-neg}.

\begin{cor}
The set of paths visitable in $\designset{E}$, noted $V_{\designset{E}}$, is positive-prefix closed:
if $\pathLL{p}\kappa^+w \in V_{\designset{E}}$ then $\pathLL{p}\kappa^+ \in V_{\designset{E}}$.
\end{cor}


\section{Incarnation expressed by means of paths}\label{sec:Incarnation}


\subsection{Incarnation}

A behaviour may contain useless designs with respect to orthogonality. This is clear when considering a design as a set of paths:  if $\behaviour{B}$ is a behaviour and $\design{D} \in \behaviour{B}$ then designs in $\behaviour{B}$ obtained from $\design{D}$ by extending one of its paths are useless with respect to membership to $\behaviour{B}$. This suggests the following definition:

\begin{defi}[Incarnation]\label{incarnationdessein}
Let $\behaviour{B}$ be a behaviour, $\design{D}$ be a design in $\behaviour{B}$. 
\begin{itemize}
\item $\Dincarnation{D}{\behaviour{B}} := \bigcup_{\design{R} \in \behaviour{B}^\perp} \normalisationDes{D}{R}$ is the {\em incarnation} of $\design{D}$ with respect to the behaviour $\behaviour{B}$. 
\item $\design{D}$ is {\em material} in $\behaviour{B}$ if $\design{D} = \Dincarnation{D}{\behaviour{B}}$.
\item The {\em incarnation} $\Bincarnation{\behaviour{B}}$ of a behaviour $\behaviour{B}$ is the  set of its material designs.
\end{itemize}
$\Dincarnation{D}{\behaviour{B}}$ is simply noted $\Dincarnation{D}{}$ when $\behaviour{B}$ is clear from the context.
\end{defi}

Note that $\bigcup_{\design{R} \in \behaviour{B}^\perp} \normalisationDes{D}{R}$ is a design: if $\design{R} \in \behaviour{B}^\perp$ then 
$\normalisationSeq{\design{D}}{\design{R}}$ is a path included in $\design{D}$, thus $\Dincarnation{D}{} = \bigcup_{\design{R} \in \behaviour{B}^\perp} \normalisationDes{D}{R}$ is a design included in $\design{D}$. 
Furthermore, by construction, $\bigcup_{\design{R} \in \behaviour{B}^\perp} \normalisationDes{D}{R} \in \behaviour{B}^{\perp\perp} = \behaviour{B}$. Hence $\Bincarnation{\behaviour{B}} \subset \behaviour{B}$.

 With respect to the inclusion relation on designs belonging to a given behaviour, the minimal designs are in the incarnation of this behaviour.
This is the core of the definition given by Girard in~\cite{DBLP:journals/mscs/Girard01}:
\begin{prop}
Let $\behaviour{B}$ be a behaviour, $\design{D}$ be a design in $\behaviour{B}$.
If $\design{D}$ is minimal in $\behaviour{B}$ with respect to inclusion then $\design{D} \in \Bincarnation{\behaviour{B}}$.
\end{prop}

Note that the previous property is also true in the more general setting defined by Basaldella and Faggian for dealing with exponentials in Ludics~\cite{DBLP:conf/lics/BasaldellaF09}. The converse property, \ie, the fact that actions can only be used once in a normalization, is not true in their setting. In the present paper, thanks to linearity, the converse property does hold,. Linearity induces the following fact: $\dual{\normalisationSeq{\design{D}}{\design{R}}} = \normalisationSeq{\design{R}}{\design{D}}$ (see also~\cite{DBLP:journals/mscs/Girard01}).

\begin{prop}
Let $\behaviour{B}$ be a behaviour, $\design{D}$ be a design in $\behaviour{B}$.
If $\design{D} \in \Bincarnation{\behaviour{B}}$
then $\design{D}$ is minimal in $\behaviour{B}$ with respect to inclusion.
\end{prop}
\begin{proof}
Remark that $\bigcup_{\design{R} \in \behaviour{B}^\perp} \normalisationDes{D}{R} = \fullview{\bigcup_{\design{R} \in \behaviour{B}^\perp} \normalisationSeq{\design{D}}{\design{R}}}$.
To obtain a design $\design{D}_0$ strictly included in $\bigcup_{\design{R} \in \behaviour{B}^\perp} \normalisationDes{D}{R}$, we have to erase at least a chronicle $\chronicle{c}$ (and its extensions). But there is at least a path $\pathLL{p}_0\in {\bigcup_{\design{R} \in \behaviour{B}^\perp} \normalisationSeq{\design{D}}{\design{R}}}$ such that $\chronicle{c}\in\fullview{\pathLL{p}_0}$   hence if we denote by $\design{R}_0$ a net such that $\pathLL{p}_0= \normalisationSeq{\design{D}}{\design{R}_0}$ we have by linearity constraint on designs that $\design{D}_0\not\perp\design{R}_0$. Hence $\design{D}_0 \not\in \behaviour{B}$.
\end{proof}

Incarnation may be also defined for a set $\designset{F}$ of nets of designs, when such a set $\designset{F}$ is the orthogonal of a set $\designset{E}$ of designs of the same base, \ie, $\designset{F} = \designset{E}^\perp$:
\begin{defi}
Let $\designset{E}$ be a set of designs of the same base $\beta=\delta_0\vdash\delta_1,\dots,\delta_n$ (resp. $\beta= \vdash\delta_1,\dots,\delta_n$), $\design{R} = (\design{E}_i)$ be a net of designs belonging to $\designset{E}^\perp$, with $\design{E}_i$ of base $ \delta_i\vdash $ for $i\in\{1,\dots, n\}$ and $\design{E}_0$ based on $\vdash\delta_0$. 
With $\design{D} \in \designset{E}$, let $(\design{F}_i^{\design{D}}) = \normalisationDes{R}{D}$ be a net of designs such that $\design{F}_i^{\design{D}}$ is a design with the same base as $\design{E}_i$.\\
The {\em incarnation} of $\design{R}$, denoted $\Dincarnation{R}{\designset{E}^{\perp}}$, is the net $(\bigcup_{\design{D} \in \designset{E}} \design{F}_i^{\design{D}})$.\\
The {\em incarnation} of $\designset{E}^{\perp}$, written $|\designset{E}^{\perp}|$, is the set of nets of designs $\{\Dincarnation{R}{\designset{E}^{\perp}} ; \design{R} \in \designset{E}^{\perp}\}$.
\end{defi}

The incarnation of such nets of designs is well-defined: with notations given in the previous definition, for all $i$, $\design{F}_i^{\design{D}} \subset \design{E}_i$, hence $\bigcup_{\design{D} \in \designset{E}} \design{F}_i^{\design{D}}$ is a design included in $\design{E}_i$.
Furthermore, for all $\design{D} \in \designset{E}$, $\design{D} \perp \Dincarnation{R}{\designset{E}^{\perp}}$. Hence $|\designset{E}^{\perp}| \subset \designset{E}^{\perp}$.
Abusively, we note $\Dincarnation{R}{\designset{E}^{\perp}} = \bigcup_{\design{D} \in \designset{E}}\normalisationDes{R}{D}$.

\vspace{.5cm}
Let us consider a set of designs $\designset{E}$ of the same base, we
gave a necessary condition in proposition~\ref{stabilite-neg} for a path to be visitable in $\designset{E}$,
 notice now that visitability is sufficient for computing the incarnation of a design in the behaviour generated by $\designset{E}$: it follows from the definition that $|\design{D}|_{\designset{E}^{\perp\perp}} = \fullview{\PoD{\design{D}} \cap V_\designset{E}}$ when $\design{D} \in \designset{E}$. 
In other words, it is not necessary to compute the behaviour of a set of designs for computing the incarnation of its designs.
This suggests the following questions. Is it possible to compute directly the incarnation of the behaviour generated by $\designset{E}$, resting only on the paths visitable in $\designset{E}$? 
Are designs in $\designset{E}$ all necessary for computing this incarnation?
We remark below that none of these questions admits an obvious answer: on the one hand the biorthogonal may contain designs that are built from parts of several designs of $\designset{E}$, on the other hand even if some kinds of designs are clearly redundant some other cases remain out of reach. 

One of the easiest kind of designs of a biorthogonal that may be clearly generated from already known designs concerns those obtained by replacing subtrees by daimons.
In fact this is part of an explicit characterization of the relation $\design{D} \preccurlyeq \design{E}$, that reads $\design{D}$ is more defined than $\design{E}$, between designs with the same base. The relation $\design{D} \preccurlyeq \design{E}$ holds whenever $\design{D}^\perp\subset\design{E}^\perp$.
This characterization is part of the {\em separation theorem}~\cite{DBLP:journals/mscs/Girard01}:
the relation $\preccurlyeq$ is a partial order and we have that $\design{D} \preccurlyeq \design{E}$ iff every chronicle $\chronicle{c} \in \design{D} - \design{E}$ can be written $\chronicle{c}'\chronicle{d}$ for a certain $\chronicle{c}'$ such that $\chronicle{c}'\daimon\in \design{E}$.
The relation $\preccurlyeq$ subsumes the inclusion relation. Indeed, we have of course that if $\design{D}\subset\design{E}$ then $\design{D} \preccurlyeq \design{E}$. 
This is already taken into account in incarnation as a material design is minimal with respect to inclusion in a behaviour.
There is another fact causing a design $\design{D}$ to be more defined than a design $\design{E}$: by substituting chronicles $w\kappa^+w'$ in $\design{D}$ by the only chronicle $w\daimon$ in $\design{E}$. As a consequence, if $\design{D}$ is in the behaviour $\designset{E}^{\perp\perp}$ then so is $\design{E}$, moreover if $\design{D}$ is in the incarnation $|\designset{E}^{\perp\perp}|$ then so is $\design{E}$.
We call such a set of designs $\design{E}$ the {\em daimon closure} of $\design{D}$ and we note it ${\design{D}}^\daimon$ (see next definition). 
Relations of inclusion and to be more defined, as well as the daimon closure of a design, are schematized in Figure~\ref{fig:daimon_closure}.

\begin{figure}\label{fig:daimon_closure}
\begin{tabular}{p{1cm}cp{.2cm}p{1cm}cp{.2cm}p{1.3cm}c}
\raisebox{5ex}{$\color{black}{\design{D}} \color{black}{\preceq} {\color{red}{\design{E}}}$:} 
&
\begin{tikzpicture}[fill opacity=0.5,scale=.7]
	\filldraw[fill=red]
		(0,0)
		-- (intersection of 0,0---1,3 and -1,2.5--1,2.5) 
		-- 	node[near start,fill opacity = 1] {{\color{red}{$\daimon$}}} 
				node[near end,fill opacity = 1] {{\color{red}{$\daimon$}}} 
				(intersection of 0,0--0.5,4 and -1,2.5--1,2.5)
		-- (0.5,4)
		-- (2.5,4)
		-- node[midway,right,fill opacity = 1] {{\color{red}{$\design{E}$}}} (0,0);

	\draw[black,line width=3pt] 
		(0,0) 
		-- node[midway,left,fill opacity = 1] {{\color{black}{$\design{D}$}}} (-1,3)
		-- (1,3)
		-- cycle;
\end{tikzpicture}
&
~
&
\raisebox{5ex}{${\color{black}{\design{D}}} \color{black}{\subset} {\color{red}{\design{E}}}$:}
&
\begin{tikzpicture}[fill opacity=0.5,scale=.7]
	\filldraw[fill=red]
		(0,0)
		-- (intersection of 0,0---1,3 and -1,3.5--1,3.5) 
		-- (intersection of 0,0--0.5,4 and -1,3.5--1,3.5)
		-- (0.5,4)
		-- (2.5,4)
		-- node[midway,right,fill opacity = 1] {{\color{red}{$\design{E}$}}} (0,0);

	\draw[black,line width=3pt] 
		(0,0) 
		-- 	node[midway,left,fill opacity = 1] {{\color{black}{$\design{D}$}}} 
				(-1,3)
		-- (1,3)
		-- cycle;
\end{tikzpicture}
&
~
&
\raisebox{5ex}{${\color{red}{\design{E}}} \color{black}{\in \design{D}^\daimon}$:}
&
\begin{tikzpicture}[fill opacity=0.5,scale=.7]
	\filldraw[fill=red]
		(0,0)
		-- (intersection of 0,0---1,3 and -1,2.5--1,2.5) 
		-- 	node[near start,fill opacity = 1] {{\color{red}{$\daimon$}}} 
				node[near end,fill opacity = 1] {{\color{red}{$\daimon$}}} 
				(intersection of 0,0--0.5,3 and -1,2.5--1,2.5)
		-- (0.5,3)
		-- (1,3)
		-- node[midway,right,fill opacity = 1] {{\color{red}{$\design{E}$}}} (0,0);

	\draw[black,line width=3pt] 
		(0,0) 
		-- node[midway,left,fill opacity = 1] {{\color{black}{$\design{D}$}}} (-1,3)
		-- (1,3)
		-- cycle;
\end{tikzpicture}
\end{tabular}
\caption{Relations of inclusion and to be more defined, and daimon closure of a design.}
\end{figure}

\begin{defi}[Daimon closure]
Let $\design{D}$ be a design, the {\em daimon closure} of $\design{D}$, denoted by ${\design{D}}^\daimon$,
is the set of designs obtained from $\design{D}$ by substituting, for some set of  negative chronicles $\chronicle{c}\in\design{D}$, all the chronicles $\chronicle{c}\kappa^+w\in\design{D}$ by the chronicles $\chronicle{c}\daimon$.\\
Let $\designset{E}$ be a set of designs of the same base, the daimon closure of $\designset{E}$ noted ${\designset{E}}^\daimon$ is the set $\bigcup_{\design{D}\in\designset{E}} \design{D}^\daimon$.\\
\end{defi}

To rephrase what precedes, behaviours and incarnation are `closed under daimon'. Hence the complement $\designset{E} \setminus \designset{E}^\daimon$ of $\designset{E}^\daimon$ in $\designset{E}$ is sufficient for generating the incarnation of $\designset{E}$. However this does not fully address the question: in example~\ref{examples_4}, designs $\design{E}$ and $\design{F}$ are sufficient for generating the incarnation of, say, $\{\design{E}, \design{F}, \design{G}, \design{H}\}$.
\begin{exa}\label{examples_3}
We have that $|\{\design{E}', \design{F}', \design{G}'\}^{\perp\perp}| = \{\design{E}', \design{F}', \design{G}', \design{E}'', \design{F}''\}^\daimon$ where:\footnote{In the representation of designs, we may choose to put $\mu$ or $\sigma$ in any of the branches when they are not a focus of an action. The choice we take has no consequence on the result.}\\
\[
\design{E}' =
\infer{\vdash \xi, \sigma, \mu}{
	\infer{\xi.1\vdash \mu}{
		\infer{\vdash \xi.1.0, \mu}{
			\xi.1.0.0 \vdash \mu
		}
	}
	&
	\infer{\xi.2\vdash \sigma}{
		\infer{\vdash \xi.2.0, \sigma}{
			\sigma.1 \vdash \xi.2.0
		}
	}
}
,
\design{F}' =
\infer{\vdash \xi, \sigma, \mu}{
	\infer{\xi.1\vdash \mu}{
		\infer{\vdash \xi.1.0, \mu}{
			\xi.1.0.1 \vdash \mu
		}
	}
	&
	\infer{\xi.2\vdash \sigma}{	
		\infer{\vdash \xi.2.0, \sigma}{
			\sigma.2 \vdash \xi.2.0
		}
	}
}
,
\design{G}' =
\infer{\vdash \xi, \sigma, \mu}{
	\infer{\xi.1\vdash \mu}{
		\infer{\vdash \xi.1.0, \mu}{
			\infer{\mu.0 \vdash \xi.1.0}{
				\infer[\daimon]{\vdash \xi.1.0, \mu.0.0}{
				}
			}
		}
	}
	&
	\xi.2\vdash \sigma
}
\]

\[
\design{E}'' =
\infer{\vdash \xi, \sigma, \mu}{
	\infer{\xi.1\vdash \mu}{
		\infer{\vdash \xi.1.0, \mu}{
			\infer{\mu.0 \vdash \xi.1.0}{
				\infer{\vdash \xi.1.0, \mu.0.0}{
					\xi.1.0.0 \vdash \mu.0.0
				}
			}
		}
	}
	&
	\infer{\xi.2\vdash \sigma}{
		\infer{\vdash \xi.2.0, \sigma}{
			\sigma.1 \vdash \xi.2.0
		}
	}
}
,
\quad\design{F}'' =
\infer{\vdash \xi, \sigma, \mu}{
	\infer{\xi.1\vdash \mu}{
		\infer{\vdash \xi.1.0, \mu}{
			\infer{\mu.0 \vdash \xi.1.0}{
				\infer{\vdash \xi.1.0, \mu.0.0}{
					\xi.1.0.1 \vdash \mu.0.0
				}
			}
		}
	}
	&
	\infer{\xi.2\vdash \sigma}{	
		\infer{\vdash \xi.2.0, \sigma}{
			\sigma.2 \vdash \xi.2.0
		}
	}
}
\]
\end{exa}

\begin{exa}\label{examples_4}
We have that $|\{\design{E}, \design{F}\}^{\perp\perp}| = \{\design{E}, \design{F}, \design{G}, \design{H}\}^\daimon$ where:
\[
\design{E}=\!\!\!\!\!\!\!\!\!\!
\infer{\vdash \xi, \sigma, \mu}{
	\infer{\xi.1\vdash \sigma, \mu}{
		\infer{\vdash \xi.1.0, \sigma, \mu}{
			\infer{\sigma.1 \vdash \xi.1.0, \mu}{
				\infer{\vdash \xi.1.0, \sigma.1.0, \mu}{
					\infer{\mu.1 \vdash \xi.1.0, \sigma.1.0}{
						\infer[\daimon]{\vdash \xi.1.0, \sigma.1.0, \mu.1.0}{
						}
					}
				}
			}
		}
	}
}
,
\design{F} =\!\!\!\!\!\!\!\!\!\!
\infer{\vdash \xi, \sigma, \mu}{
	\infer{\mu.1\vdash \sigma, \xi}{
		\infer{\vdash \mu.1.0, \sigma, \xi}{
			\infer{\sigma.1 \vdash \mu.1.0, \xi}{
				\infer{\vdash \mu.1.0, \sigma.1.0, \xi}{
					\infer{\xi.1 \vdash \mu.1.0, \sigma.1.0}{
						\infer[\daimon]{\vdash \xi.1.0, \sigma.1.0, \mu.1.0}{
						}
					}
				}
			}
		}
	}
}
,
\design{G} =\!\!\!\!\!\!\!\!\!\!
\infer{\vdash \xi, \sigma, \mu}{
	\infer{\xi.1\vdash \sigma, \mu}{
		\infer{\vdash \xi.1.0, \sigma, \mu}{
			\infer{\mu.1 \vdash \xi.1.0, \sigma}{
				\infer{\vdash \xi.1.0, \mu.1.0, \sigma}{
					\infer{\sigma.1 \vdash \xi.1.0, \mu.1.0}{
						\infer[\daimon]{\vdash \xi.1.0, \sigma.1.0, \mu.1.0}{
						}
					}
				}
			}
		}
	}
}
,
\design{H} =\!\!\!\!\!\!\!\!\!\!
\infer{\vdash \xi, \sigma, \mu}{
	\infer{\mu.1\vdash \sigma, \xi}{
		\infer{\vdash \mu.1.0, \sigma, \xi}{
			\infer{\xi.1 \vdash \mu.1.0, \sigma}{
				\infer{\vdash \mu.1.0, \xi.1.0, \sigma}{
					\infer{\sigma.1 \vdash \mu.1.0, \xi.1.0}{
						\infer[\daimon]{\vdash \xi.1.0, \sigma.1.0, \mu.1.0}{
						}
					}
				}
			}
		}
	}
}
\]
\end{exa}

Note that each design is obtained from each other one by commuting some actions.
Obviously, such commutations do not always give rise to elements of the biorthogonal as it depends on the rest of the set of designs. Moreover, designs in the incarnation may not necessarily be obtained this way: In example~\ref{examples_3}, designs $\design{E}''$ and $\design{F}''$ are built by adding the pair of actions $(+,\mu,\{0\})(-,\mu.0,\{0\})$ into paths which become visitable reversing proposition~\ref{stabilite-neg}: every path in a dual design travels through this pair.

\vspace{1em}

\vspace{.5cm}
From the previous remarks, it seems that there is no clear way to compute directly an incarnation. To circumvent this difficulty, we use in this paper an indirect approach for computing an incarnation. 
First, as we develop it in the next subsection, the incarnation $|\designset{E}^\perp|$ may be defined by means of $\designset{E}$ only, where $\designset{E}$ is a set of designs.
Second we remark with the following proposition that $|\designset{E}^{\perp\perp}| = |(|\designset{E}^\perp|)^\perp|$. Hence computing an incarnation involves twice the same procedure, \ie, computing the incarnation of the dual.
\begin{prop}\label{prop:charac_incarnation_dualdual}
Let $\designset{E}$ be a set of designs of the same base $\beta$,
$\designset{E}^{\perp\perp} = |\designset{E}^{\perp}|^\perp$.
\end{prop}
\begin{proof}
We already noticed that $|\designset{E}^{\perp}| \subset \designset{E}^{\perp}$. Hence $\designset{E}^{\perp\perp} \subset |\designset{E}^{\perp}|^\perp$.
Let $\design{D} \in |\designset{E}^{\perp}|^\perp$, for all $\design{R} \in \designset{E}^{\perp}$, $\design{D} \perp \Dincarnation{R}{\designset{E}^{\perp}}$ hence $\normalisationSeq{\Dincarnation{R}{\designset{E}^{\perp}}}{\design{D}}$ is a normalization path in $\Dincarnation{R}{\designset{E}^{\perp}}$, hence also in $\design{R}$. Thus $\design{D} \perp \design{R}$. 
\end{proof}

\subsection{Direct computation of the incarnation of the dual of a set of designs}
 
Let us define
$\dual{V_\designset{E}} = \{\dual{\pathLL{p}} ; \pathLL{p} \in V_\designset{E}\}$. It is obvious that if $\pathLL{p}$ is visitable in $\designset{E}$ then $\dual{\pathLL{p}}$ is visitable in $\designset{E}^\perp$. In other words, $\dual{V_\designset{E}} \subset V_{\designset{E}^\perp}$. 
Recall that a design in $|\designset{E}^\perp|$ should be the set of views of coherent paths of $\dual{V_\designset{E}}$. Moreover, such a design should be orthogonal to each design in $\designset{E}$, hence
maximal cliques of $\dual{V_\designset{E}}$ are natural candidates for defining designs of $|\designset{E}^\perp|$. 
This is not sufficient. Indeed we get a counter-example by transposing in Ludics an example used by Ehrhard for the study of hypercoherences~\cite{DBLP:journals/tcs/Ehrhard00}. In that counter-example~\ref{exa:anticlique}, we notice that the infinite sequence resulting   of an increasing sequence of visitable paths has to belong to some design of $\designset{E}$.
With that constraint we are able to characterize the incarnation of the dual (propositions~\ref{prop:clique_dual} and~\ref{prop:cs_clique_dual}).
First of all, we remark with the following proposition that the dual of a clique of visitable paths is an anticlique of paths, this property is necessary for proving materiality.

\begin{prop}
Let $\designset{E}$ be a set of designs based on $\beta$, $\pathLL{p}$ and $\pathLL{q}$ be two distinct paths of $V_\designset{E}$,
\begin{itemize}
\item if $\pathLL{p} \coh \pathLL{q}$ then $\dual{\pathLL{p}} \not\coh \dual{\pathLL{q}}$,
\item if $C$ is a clique of paths based on $\beta$ visitable in $\designset{E}$ then $\dual{C}$ is an anticlique of paths based on $\beta^\perp$.
\end{itemize}
\end{prop}

\begin{proof}
Let $\pathLL{p}, \pathLL{q}$ be two distinct paths of $V_\designset{E}$, and suppose that $\pathLL{p} \coh \pathLL{q}$:
\begin{itemize}
\item Either one sequence strictly extends the other: w.l.o.g. $\pathLL{p} = \pathLL{q}\kappa^-w$. 
Then $\dual{\pathLL{q}} = \overline{\pathLL{q}}\daimon$ and $\dual{\pathLL{p}} = \overline{\pathLL{q}}\overline{\kappa^-}\dual{w}$. The paths $\dual{\pathLL{p}}$ and $\dual{\pathLL{q}}$ are strictly incoherent.
\item Or there exists two distinct negative actions $\kappa_1^-$ and $\kappa_2^-$ such that $\pathLL{p} = w\kappa_1^-w_1$ and $\pathLL{q} = w\kappa_2^-w_2$. 
Then $\dual{\pathLL{p}} = \overline{w}\overline{\kappa_1^-}\dual{w_1}$ and $\dual{\pathLL{q}} = \overline{w}\overline{\kappa_2^-}\dual{w_2}$  which are strictly incoherent. 
\end{itemize}
It follows that if $C$ is a clique of $V_\designset{E}$ then $\dual{C}$ is an anticlique of paths.
\end{proof}

Let $C$ be a set of visitable paths. The fact that $\dual{C}$ is an anticlique of paths (a set of pairwise incoherent paths) does not ensure that $C$ is a clique, as the following example shows.
\begin{exa}\label{exa:anticlique}
Let $\designset{E}=\{\design{E},\design{F}\}$ where:\\
\begin{center}
$\design{E}=\infer{\vdash \xi}{
	\infer{\xi1 \vdash}{
		\infer{\vdash \xi11}{
			\xi110 \vdash
		}
	}
	&
	\infer{\xi2 \vdash}{
		\infer{\vdash \xi21}{
			\xi210 \vdash
		}
	}
}$
~~~and~~~
$\design{F}=\infer{\vdash \xi}{
	\infer{\xi1 \vdash}{
		\infer{\vdash \xi11}{
			\xi110 \vdash
		}
	}
	&
	\infer{\xi2 \vdash}{
		\infer{\vdash \xi21}{
			\xi211 \vdash
		}
	}
}$
\end{center}
Paths $\pathLL{p}$ 
and $\pathLL{q}$ defined below are incoherent positive visitable paths in $\designset{E}$ and $\dual{\pathLL{p}}$ and $\dual{\pathLL{q}}$ are also incoherent:\\
$\pathLL{p}=(+,\xi,\{1,2\})(-,\xi1,\{1\})(+,\xi11,\{0\})(-,\xi2,\{1\})(+,\xi21,\{0\})$\\
$\pathLL{q}=(+,\xi,\{1,2\})(-,\xi2,\{1\})(+,\xi21,\{1\})$.
\end{exa}

We give below a direct characterization of the incarnation of the dual of a set of designs $\designset{E}$.
Though simple, it has the disadvantage of requiring a test toward each design of the set $\designset{E}$. We develop in propositions~\ref{prop:clique_dual} and~\ref{prop:cs_clique_dual} a more complex characterization that considers only tests toward the set of visitable paths of $\designset{E}$.

\begin{prop}
Let $\designset{E}$ be a set of designs based on $\beta$. Let $C \subset V_\designset{E}$ such that for all design $\design{D} \in \designset{E}$, $C \cap \PoD{\design{D}} \neq \emptyset$ and $\dual{C}$ is a clique of $\dual{V_\designset{E}}$, then $\fullview{\dual{C}}$ belongs to $|\designset{E^\perp}|$.
\end{prop}
\begin{proof}
-- $\dual{C}$ is a clique then $\fullview{\dual{C}}$ is a design.\\
-- As normalisation is deterministic and by construction of $\fullview{\dual{C}}$, we have that for all design $\design{D} \in \designset{E}$, $\design{D} \perp \fullview{\dual{C}}$, \ie, $\fullview{\dual{C}} \in \designset{E}^\perp$.\\
-- Suppose there exists a strict subnet $\design{R}$ of $\fullview{\dual{C}}$ that is in $\designset{E}^\perp$. Hence there exists a path $\dual{\pathLL{p}} \in \dual{C}$ and $\dual{\pathLL{p}} \not\in \PoD{\design{R}}$. By construction of $\fullview{\dual{C}}$, there exists a design $\design{D} \in \designset{E}$ such that $\pathLL{p} \in \PoD{\design{D}}$. Hence normalisation between $\design{R}$ and $\design{D}$ fails : $\design{R} \not\in \designset{E}^\perp$. Contradiction. Then $\fullview{\dual{C}} \in |\designset{E^\perp}|$.
\end{proof}

\begin{prop}
Let $\designset{E}$ be a set of designs based on $\beta$. 
Let $\design{R} \in |\designset{E}^\perp|$ then there exists $C \subset V_\designset{E}$ such that for all design $\design{D} \in \designset{E}$, $C \cap \PoD{\design{D}} \neq \emptyset$, $\dual{C}$ is a clique of $\dual{V_\designset{E}}$ and $\fullview{\dual{C}} = \design{R}$.
\end{prop}
\begin{proof}
As $\design{R} \in |\designset{E}^\perp|$, we know that $\design{R} = \fullview{\{\normalisationSeq{\design{R}}{\design{D}} ; \design{D} \in \designset{E}\}}$. Let us consider $C=\dual{\dual{C}}$ where $\dual{C} = \{\normalisationSeq{\design{R}}{\design{C}} ; \design{C} \in \designset{E}\}$.
\begin{itemize}
\item By construction, $\dual{C}$ is a subset of $\dual{V_\designset{E}}$ and for all design $\design{D} \in \designset{E}$, $C \cap \PoD{\design{D}} \neq \emptyset$. 
\item As $\design{R}$ is a design and $\dual{C} \subset\PoD{\design{R}}$, $\dual{C}$ is a clique.\qedhere
\end{itemize}
\end{proof}

Let us consider now what constraints should be required on such sets $C$ referring only to visitable paths for $C$ to give rise to a design of the incarnation of the dual.

To be a maximal clique of $\dual{V_{\designset{E}}}$ is not sufficient for characterizing the material designs of $\designset{E}^\perp$, when $\designset{E}$ is a set of designs. 
Two additional constraints are required.
First, such a set of paths $C$ should be {\em saturated} with respect to paths of $V_{\designset{E}}$: if a visitable path can be followed in a design $\design{D}$ of $\designset{E}$ and dually in $\dual{C}$ then the next positive action in $\design{D}$ should be present dually in the path of $\dual{C}$, except when this positive action is a daimon.
Second, the intersection of the anticlique made of duals with designs of $\designset{E}$ should be, in some way, {\em finite}.
Indeed, the same phenomenon also happens with serial-parallel coherence spaces (the ones obtained from $1= \{*\}$ with the connectives $\with$ and $\oplus$): being given a family of maximal cliques that covers a coherent space, maximal anticliques cut each such maximal clique if the space is finite, this is not always the case if the space is infinite. Example~\ref{example-infini} given below is a transposition in Ludics of the one given by T. Ehrhard in his study of hypercoherences~\cite{DBLP:journals/tcs/Ehrhard00}.

\begin{exa}\label{example-infini}


Let us consider designs generated by the following grammar ($n \geq 0$):\\
\begin{center}
$\begin{array}{ccccccc}
\design{D}^0_\xi := 
\raisebox{-2.2mm}{
\infer{\xi \vdash}{
	\infer{\vdash \xi0}{\xi00 \vdash}
}
}
& \hspace{1cm} &
\design{D}^1_\xi := 
\raisebox{-2.2mm}{
\infer{\xi \vdash}{
	\infer{\vdash \xi1}{\xi10 \vdash}
}
}
& \hspace{1cm} &
\design{D}^{2n+2}_\xi := 
\raisebox{-2.2mm}{
\infer{\xi \vdash}{
	\infer{\vdash \xi1}{
		\infer*[\design{D}^{2n}_{\xi11}]{\xi11\vdash}{}
	}
}
}
& \hspace{1cm} &
\design{D}^{2n+3}_\xi := 
\raisebox{-2.2mm}{
\infer{\xi \vdash}{
	\infer{\vdash \xi1}{
		\infer*[\design{D}^{2n+1}_{\xi11}]{\xi11\vdash}{}
	}
}
}
\end{array}
$
\end{center}
\vspace{0,5em}

\noindent For all $n \geq 0$, we note $\underline{n}$ the maximal chronicle of the design $\design{D}_\xi^n$. We note $\underline{\dual{n}}$ the dual of $\underline{n}$. 
Remark that, for all $n,p,k \geq 0$, 
\[
\underline{2n} \coh \underline{2p}, \quad\underline{2n} \coh \underline{2(n+k)+1}\mbox{ and, if }n \neq p, \quad\underline{2n+1} \not\coh \underline{2p+1}.
\]
Hence for all $n \geq 0$, 
\[
\{\underline{0},\underline{2},\dots,\underline{2n},\underline{2n+1}\}^*\mbox{ is a design as well as }\{\underline{0},\underline{2},\underline{4},\dots,\underline{2n},\dots\}^*.
\]
Furthermore for all $n \geq 0$, $\{\underline{\dual{1}},\underline{\dual{3}},\dots,\underline{\dual{2n+1}},\underline{\dual{2n+2}}\}^*$ is a design as well as $\{\underline{\dual{0}}\}^*$ and $\{\underline{\dual{1}},\underline{\dual{3}},\dots,\underline{\dual{2n+1}},\dots\}^*$.\\
\vspace{1em}%
For example, 
$
\{\underline{0},\underline{1}\}^* = 
\raisebox{-2.2mm}{
\infer{
	{\xi \vdash}
}{
	\infer{
		{\vdash \xi 0}
	}{
		{\xi00 \vdash }
	}
	&
	\infer{
		{\vdash \xi1}
	}{
		{\xi10 \vdash }
	}
}
}
$
, 
$
\{\underline{0},\underline{2},\underline{3}\}^* = 
\raisebox{-2.2mm}{
\infer{
	{\xi \vdash}
}{
	\infer{
		{\vdash \xi0}
	}{
		{\xi00 \vdash }
	}
	&
	\infer{
		{\vdash \xi1}
	}{
		\infer {
			{\xi11 \vdash }
		}{
			\infer{
				{\vdash \xi110}
			}{
				{\xi1100 \vdash }
			}
			&
			\infer{
				{\vdash \xi111}
			}{
				{\xi1110 \vdash}
			}
		}
	}
}
}
$
and
\begin{flushright}
$
\{\underline{\dual{0}}\}^* = 
\raisebox{-2.2mm}{
\infer{
	{\vdash \xi}
}{
	\infer{
		{\xi0 \vdash }
	}{
		\infer[\daimon]
			{\vdash \xi00}
			{}
	}
}
}
$
,
$
\{\underline{\dual{1}},\underline{\dual{3}},\underline{\dual{4}}\}^* = 
\raisebox{-2.2mm}{
\infer{
	{\vdash \xi}
}{
	\infer{
		{\xi1 \vdash}
	}{
		\infer[\daimon]{\vdash \xi10}{}
		&
		\infer {
			{\vdash \xi11}
		}{
			\infer{
				{\xi111 \vdash}
			}{
				\infer[\daimon]{\vdash \xi1110}{}
				&
				\infer{
					{\vdash \xi1111}
				}{
					\infer{
						\xi11110 \vdash 
						}{
						\infer[\daimon]{\vdash \xi111100}{}
						}
				}
			}
		}
	}
}
}
$
\end{flushright}
Consider the two following designs:
\[
\design{E} = \{\underline{0},\underline{2},\underline{4},\dots,\underline{2n},\dots\}^*, 
\quad \design{F} = \{\underline{\dual{1}},\underline{\dual{3}},\dots,\underline{\dual{2n+1}},\dots\}^*,
\]
and the two following sets of designs:
\[
 \designset{E} = \{\{\underline{0},\underline{2},\dots,\underline{2n},\underline{2n+1}\}^* ; \forall n \geq 0\}, \quad\designset{F} = \{\{\underline{\dual{0}}\}^*\} \cup \{\{\underline{\dual{1}},\underline{\dual{3}},\dots,\underline{\dual{2n+1}},\underline{\dual{2n+2}}\}^* ; \forall n \geq 0\}.
\]
Observe that $V_\designset{E} = V_{\{\design{E}\} \cup \designset{E}}$ is the set of positive prefixes of paths $\underline{n}$ for all $n \geq 0$.
We have that $\{\design{F}\} \cup \designset{F} \subset |\designset{E}^\perp|$ and $\designset{F} \subset |(\{\design{E}\} \cup \designset{E})^\perp|$. 
The remainders of $|\designset{E}^\perp|$ and $|(\{\design{E}\} \cup \designset{E})^\perp|$ are obtained from designs of $\designset{F}$ by replacing proper positive chronicles $\chronicle{c}$ and their extensions by chronicles $\dual{\chronicle{c}}$. 
The sets of views of maximal cliques of $\dual{V_\designset{E}}$ are then exactly material designs of $\designset{E}$.
We note that $\design{F} \not\in |(\{\design{E}\} \cup \designset{E})^\perp|$.
Indeed, $\design{E}$ is not orthogonal to $\design{F}$ since their interaction does not end.
In other words, this is so because $(+,\xi,\{1\}), (-,\xi.1,\{1\}), \dots$ is an infinite sequence of actions in $\design{F}$ whose dual is included in $\design{E}$.
Excepting this case, the other sets of views of maximal cliques of $\dual{V_{\{\design{E}\} \cup \designset{E}}}$ are exactly material designs of $\{\design{E}\} \cup \designset{E}$.
\end{exa}

\begin{defi}
Let $\designset{E}$ be a set of designs based on $\beta$. Let $C$ be a
set of paths of designs of $\designset{E}$.
\begin{itemize}[label=$-$]
\item The set $C$ is {\bf finite-stable} when 
for all strictly increasing sequence $(\pathLL{p}_n)$ of elements of $C$, 
if $\bigcup \fullview{\pathLL{p}_n}$ is included in a design of $\designset{E}$ 
then the sequence $(\pathLL{p}_n)$ is finite.
\item The set $C$ is {\bf saturated} when
	 for all $\pathLL{q}$ prefix of an element of $C$ and
         $\pathLL{q}\kappa^+ \in V_\designset{E}$ with $\kappa^+ \neq
         \daimon$, we have that $\pathLL{q}\kappa^+$ is a prefix of an
         element of $C$.
\end{itemize}
\end{defi}

\begin{prop}\label{prop:clique_dual}
Let $\designset{E}$ be a set of designs based on $\beta$. Let $C \subset V_\designset{E}$ such that $C$ is finite-stable, saturated and $\dual{C}$ is a maximal clique of $\dual{V_\designset{E}}$. Then $\fullview{\dual{C}}$, the set of views of paths in $\dual{C}$, is a design that belongs to $|\designset{E^\perp}|$.
\end{prop}
\begin{proof}
Suppose that $V_\designset{E} = \{\epsilon\}$, then $\designset{E}^\perp = \{\dai\}$ and $\dual{V_\designset{E}} = \{\daimon\}$ and the result follows.
In the following we consider that there are some non empty paths visitable in $\designset{E}$.
As $\dual{C}$ is a clique of $\dual{V_\designset{E}}$, $\fullview{\dual{C}}$ is a net of designs.
\begin{itemize}
\item The net $\fullview{\dual{C}}$ is in $\designset{E}^\perp$:\\
Let $\design{E}\in\designset{E}$. We prove that there is a path $\pathLL{q}\in \PoD{\design{E}} \cap C$. This way we prove that  $\fullview{\dual{C}}\perp\design{E}$: either $\pathLL{q}$ or $\dual{\pathLL{q}}$ ends with a daimon hence the interaction finishes well. We note $V_\design{E}^C$ the set of paths in $\PoD{\design{E}}$ that end with a positive action and  are prefixes of paths in $C$ (hence visitable in $\designset{E}$).
\begin{itemize}
\item Suppose that $V_\design{E}^C$ is not empty. Because $C$ is finite-stable and $\design{E} \in \designset{E}$, sequences of elements of $V_\design{E}^C$ cannot be infinite strictly increasing. Hence we may consider a path $\pathLL{q}$ maximal in $V_\design{E}^C$, \ie, maximal such that $\pathLL{q} \in \PoD{\design{E}}$ ends with a positive action and is a prefix of an element of $C$. Note that if $\pathLL{q}$ ends with a daimon then $\pathLL{q} \in C$.
Either $\pathLL{q}\in C$, hence $\pathLL{q} = \,<\!\!{\design{E}}\leftarrow{\fullview{\dual{C}}}\!\!>$, thus $\design{E}\perp \fullview{\dual{C}}$. 
Or there are a negative action $\kappa^-$ and a sequence $w$ such that $\pathLL{q}\kappa^-w\in C$. Hence, from proposition~\ref{stabilite-neg}, $\pathLL{q}\kappa^- \in \PoD{\design{E}}$.
Moreover $\pathLL{q}\kappa^-$ has an extension $\pathLL{q}\kappa^-\kappa^+$ in $\PoD{\design{E}}$. 
We construct now a net in $\designset{E^\perp}$ that `forces' the interaction with $\pathLL{q}\kappa^-$:

\begin{tabular}{lcl}
let $P$ & $=$ & $\{\pathLL{p}\kappa'^-\daimon ~; \pathLL{p}$ is a prefix of $\overline{\pathLL{q}}$ and $\pathLL{p}\kappa'^-$ is not a prefix of $\overline{\pathLL{q}}\}$ \\
	&$\cup$& $\{\overline{\pathLL{q}\kappa^-}\kappa'^-\daimon ~; \pathLL{q}\kappa^-\overline{\kappa'^-}$ is a path of a design in $\designset{E}\}$
\end{tabular}

 and consider the net $\design{R}=\fullview{P}$.
We have that $\design{R}$ is in $\designset{E^\perp}$: $\overline{\pathLL{q}\kappa^-} \in \dual{V_{\designset{E}}}$, hence a normalisation path between a design of $\designset{E}$ and $\design{R}$ either follows $\pathLL{q}\kappa^-$ or diverges from it on a negative action, in the two cases interaction finishes well as $\design{R}$ is ``completed'' by means of all possible negative actions, each of them followed by a daimon. 
Then $\normalisationSeq{\design{E}}{\design{R}}=\pathLL{q}\kappa^-\kappa^+$, hence $\pathLL{q}\kappa^-\kappa^+ \in V_\designset{E}$.\\
If $\kappa^+ = \daimon$ then $\pathLL{q}\kappa^-\daimon \in \PoD{\design{E}}$ and $\overline{\pathLL{q}}\overline{\kappa^-} \in \PoD{\fullview{\dual{C}}}$, thus $\design{E}\perp \fullview{\dual{C}}$. 
Otherwise, as $C$ is saturated, $\pathLL{q}\kappa^-\kappa^+$ is a prefix of an element of $C$  then we have a contradiction with the fact that $\pathLL{q}$ is maximal in $V_\design{E}^C$.

\item We prove now that $V_\design{E}^C$ is not empty.\\
Suppose that the base of $\designset{E}$ is positive. Let $\kappa^+$ be the first action of $\design{E}$, the path $\kappa^+$ is visitable in $\designset{E}$.
If there is no path in $\dual{C}$ beginning with this action $\overline{\kappa^+}$, the path $\overline{\kappa^+}\daimon$ is coherent with all the paths of $\dual{C}$ against the maximality of $\dual{C}$. 
Hence either $\overline{\kappa^+}\daimon \in \dual{C}$ or  there is a  path $\overline{\kappa^+\kappa^-}w$ in $\dual{C}$ extending $\overline{\kappa^+}$ and $\kappa^+ \in V_\design{E}^C$. \\
 Suppose now that the base of $\designset{E}$ is negative.
 Either $\dual{C} = \{\daimon\}$ and the result follows, or all paths in $\dual{C}$ begin with the same proper positive action $\overline{\kappa^-}$. 
Hence there are sequences $w$ such that $\kappa^-w\in V_\designset{E}$ and then, for all $\design{D}\in\designset{E}$, $\kappa^-\in\PoD{\design{D}}$, in particular $\kappa^-\in\PoD{\design{E}}$. 
If $\kappa^-\daimon \in \PoD{\design{E}}$, then $\design{E}\perp \fullview{\dual{C}}$. 
Otherwise there is a unique action $\kappa^+$ such that $\kappa^-\kappa^+ \in \PoD{\design{E}}$. Furthermore, $\kappa^-\kappa^+ \in V_\designset{E}$: since $\overline{\kappa^-} \in \overline{V_\designset{E}}$ there is at least a net in $\designset{E}^\perp$ containing $\overline{\kappa^-}$ as unique first action and then, since this net is orthogonal to $\design{E}$ it has to contain a path extending $\overline{\kappa^-\kappa^+}$. 
If there is no path in $\dual{C}$  extending  $\overline{\kappa^-}$ with the action ${\overline{\kappa^+}}$, the path $\overline{\kappa^-\kappa^+}\daimon$ is coherent with all the paths of $\dual{C}$ against the maximality of $\dual{C}$. Or there is a (proper) path $\overline{\kappa^-\kappa^+}w$ in $\dual{C}$ extending $\overline{\kappa^-\kappa^+}$ and $\kappa^-\kappa^+ \in V_\design{E}^C$. 
\end{itemize}

\item The net $\fullview{\dual{C}}$ is material in $\designset{E}^\perp$.\\
Suppose there exists a strict subnet $\design{R}$ of $\fullview{\dual{C}}$ that is in $\designset{E}^\perp$. Hence there exists a path $\dual{\pathLL{p}} \in \dual{C}$ and $\dual{\pathLL{p}} \not\in \PoD{\design{R}}$. As $\pathLL{p} \in V_\designset{E}$, there exists a design $\design{E} \in \designset{E}$ and $\pathLL{p} \in \PoD{\design{E}}$. Note that there exists a unique path in $\PoD{\design{E}} \cap C$ as $C$ is an anticlique and $\design{E}$ is a design. As $\design{R} \in \designset{E}^\perp$, $\pathLL{q} = \normalisationSeq{\design{E}}{\design{R}} \in V_\designset{E}$.
Hence $\pathLL{q} \in \PoD{\design{E}}$. Moreover, $\dual{\pathLL{q}} \in \PoD{\design{R}}$, thus $\fullview{\,\dual{\pathLL{q}}\,} \subset \fullview{\dual{C}}$, hence coherent with paths in $\dual{C}$, hence $\dual{\pathLL{q}} \in \dual{C}$ as $\dual{C}$ is a maximal clique. But $\pathLL{p} \neq \pathLL{q}$ and the two paths belong to $\PoD{\design{E}} \cap C$. Contradiction.\qedhere
\end{itemize}
\end{proof}

The next proposition shows that one exactly recovers $|\designset{E}^\perp|$, the incarnation of $\designset{E}^\perp$, from the set of maximal cliques of $\dual{V_\designset{E}}$ which are dual to finite-stable subsets of $V_\designset{E}$.

\begin{prop}\label{prop:cs_clique_dual}
Let $\designset{E}$ be a set of designs of the same base. 
Let $\design{R} \in |\designset{E}^\perp|$ then there exists $C \subset V_\designset{E}$ such that $C$ is finite-stable and saturated and $\dual{C}$ is a maximal clique of $\dual{V_\designset{E}}$ and $\fullview{\dual{C}} = \design{R}$.
\end{prop}
\begin{proof}
As $\design{R} \in |\designset{E}^\perp|$, we know that $\design{R} = \fullview{\{\normalisationSeq{\design{R}}{\design{D}} ; \design{D} \in \designset{E}\}}$. Let us consider $C=\dual{\dual{C}}$ where $\dual{C} = \{\normalisationSeq{\design{R}}{\design{D}} ; \design{D} \in \designset{E}\}$.
\begin{itemize}
\item By construction $\dual{C}$ is a subset of $\dual{V_\designset{E}}$. 
\item As $\design{R}$ is a design and $\dual{C} \subset\PoD{\design{R}}$, $\dual{C}$ is a clique.
\item Suppose $\dual{C}$ is not a maximal clique of $\dual{V_\designset{E}}$: there exists a path $\pathLL{p} \in \dual{V_\designset{E}}$ such that $\forall \pathLL{q} \in \dual{C}, \pathLL{p} \coh \pathLL{q}$ and $\pathLL{p} \not\in \dual{C}$. The path $\pathLL{p} \in \dual{V_\designset{E}}$ thus there exists $\design{D} \in \designset{E}$ and $\design{R}' \in \designset{E}^\perp$ and $\pathLL{p} = \normalisationSeq{\design{R}'}{\design{D}}$. Moreover there exists $\pathLL{q} \in \dual{C}$ and $\pathLL{q} = \normalisationSeq{\design{R}}{\design{D}}$. As $\pathLL{p} \coh \pathLL{q}$, then $\dual{\pathLL{p}} \not\coh \dual{\pathLL{q}}$. But these two paths $\dual{\pathLL{p}}$ and $\dual{\pathLL{q}}$ belong to the same design $\design{D}$. Contradiction. Then $\dual{C}$ is a maximal clique of $\dual{V_\designset{E}}$.
\item Suppose that there exists a strictly increasing sequence $(\pathLL{p}_n)$ of prefixes of paths of $C$ and $\design{D} \in \designset{E}$ such that $\bigcup_n \fullview{\pathLL{p}_n} \subset \design{D}$. 
Note that $(\overline{\pathLL{p}_n})$ is a sequence of prefixes of paths of $\dual{C}$. If the sequence is infinite strictly increasing, the normalization between $\views{\dual{C}}$ and $\design{D}$ does not finish, contradicting the fact that normalization ends between $\design{R}$ and $\design{D}$.
\item $C$ is saturated: Suppose that $\pathLL{q}$ is a prefix of an element of $C$ and $\pathLL{q}\kappa^+ \in V_\designset{E}$ with $\kappa^+ \neq \daimon$. There exists a design $\design{E} \in \designset{E}$ such that $\pathLL{q}\kappa^+ \in \PoD{\design{E}}$. Furthermore $\overline{\pathLL{q}} \in \PoD{\design{R}}$. As normalisation is deterministic and $\design{R} \perp \design{E}$, there exists a sequence $w$ such that $\normalisationSeq{\design{E}}{\design{R}} = q\kappa^+w$. Hence $\pathLL{q}\kappa^+$ is a prefix of an element of $C$.\qedhere
\end{itemize}
\end{proof}

\begin{prop}\label{prop:incarnation_by_dual}
Let $\designset{E}$ be a set of designs based on $\beta$. The incarnation of the behaviour generated by $\designset{E}$ is determined by applying the following steps:
\begin{enumerate}
\item Consider the set of paths of $\designset{E}$ which may be the visitable paths of ${\designset{E}}$ (by means of proposition \ref{stabilite-neg}).
\item Establish $V_{\designset{E}}$, the set of visitable paths of $\designset{E}$, by keeping from the previous set the paths passing the proposition~\ref{prop:visitable-carac} .
\item Obtain $|\designset{E}^\perp|$ from the set of maximal cliques $\dual{C}$ of $\dual{V_\designset{E}}$ such that $C$ is finite-stable and saturated.
\item Establish $V':=V_{|\designset{E}^\perp|}$  (again by means of propositions \ref{stabilite-neg} and~~\ref{prop:visitable-carac}), the set of visitable paths of $|\designset{E}^\perp|$.
\item Obtain $|\designset{E}^{\perp\perp}|$ from the set of maximal cliques $\dual{C'}$ of $\dual{V'}$ such that $C'$ is finite-stable and saturated.
\end{enumerate}
\end{prop}
\begin{proof}
Propositions~\ref{stabilite-neg},~\ref{prop:visitable-carac},~\ref{prop:clique_dual} and \ref{prop:cs_clique_dual} allow for computing the incarnation of the {\em dual} of a set of designs. Hence the incarnation of the behaviour generated by $\designset{E}$ is established by using twice this process (proposition~\ref{prop:charac_incarnation_dualdual}).
\end{proof}

 \rem Our construction freely captures the part of incarnation corresponding to a ``daimon closure'': 
$\dual{V_\designset{E}}$ is already closed under daimon as $V_\designset{E}$  is positive-prefix closed.


\section{Conclusion}

Incarnation is an original concept introduced in Ludics by Girard. Indeed the denotation of a formula is a behaviour, a set of designs. The incarnation of a behaviour is the subset of designs that are fully used by interaction. In this paper, we addressed a more general question: is it possible to compute the incarnation of the behaviour generated by a set $\designset{E}$ of designs of the same base, without necessarily computing explicitly this behaviour?
What we could rephrase by saying that we want to isolate, among the information contained in the designs of $\designset{E}$, elementary components that it would have been enough to reorganize to calculate this incarnation. 
Our characterization of what is a ``visitable path in a set of designs'' is a key ingredient for that purpose (propositions~\ref{stabilite-neg} and ~\ref{prop:visitable-carac}). 
From this result, we were able to characterize designs in the incarnation of the orthogonal $\designset{E}^\perp$ as maximal cliques of visitable paths such that their dual are finite-stable and saturated (propositions~\ref{prop:clique_dual} and~\ref{prop:cs_clique_dual}).
We deduced then the incarnation as a twofold process (proposition~\ref{prop:incarnation_by_dual}).

In fact, the notion of incarnation suggests implicitly and indirectly the hope of finding a kind of ``basis'' (as defined in vector spaces): the incarnation of a behaviour contains enough information to find all its elements, it also avoids some redundancy since it contains only the minimal designs with respect to inclusion. Can we go further? A first step is easy to cross: we could consider a concept a little ``finer'' than that of incarnation by taking only designs minimal with respect to the relation $\preccurlyeq$ recalled in section~\ref{sec:Incarnation} and that serves for the separation theorem. In fact the set $\{\design{D} \in \designset{E}^{\perp\perp} ; \forall \design{E} \in \designset{E}^{\perp\perp}$ if $\design{E} \preccurlyeq \design{D}$ then $\design{E} = \design{D}\}$ is a subset of $|\designset{E}^{\perp\perp}|$. 
And it suffices to take its closure by daimon to find the usual incarnation.

\bibliographystyle{plain}
\bibliography{Ludics,GameSemanticsLL,GdI}

\end{document}

%% file: STYLES/STYLES.tex
\usepackage[applemac]{inputenc}
\usepackage{amsfonts,amsmath,latexsym,amssymb,amstext}
\usepackage{eufrak}
\usepackage{cmll}

\usepackage{STYLES/proofCF}
\usepackage{tabularx}

\def\eg{{\em e.g.}}
\def\ie{{\em i.e.}}

\newcommand{\Rouge}[1]{
{\color{black}#1}%
}
\newcommand{\Bleu}[1]{\protect}

\newcommand{\un}{{\mathbf 1}}
\newcommand{\zero}{{\mathbf 0}}

\newcommand{\daimon}{{\scriptstyle \maltese}}
\newcommand{\chronicle}[1]{{\mathfrak{#1}}}
\newcommand{\design}[1]{{\mathfrak{#1}}}

\newcommand{\designset}[1]{{\mathrm{#1}}}
\newcommand{\behaviour}[1]{{\mathbf{#1}}}

\newcommand{\Bincarnation}[1]{|{\designset{#1}}|}

\newcommand{\Dincarnation}[2]{|{\design{#1}}|_{\designset{#2}}}

\newcommand{\normalisation}[1]{[\![#1]\!]}
\newcommand{\psdes}[2]{[\![ #1 , #2 ]\!]}

\newcommand{\normalisationSeq}[2]{\left<#1\!\!\leftarrow\!\!#2\right>}
\newcommand{\normalisationDes}[2]{\fullview{\normalisationSeq{\design{#1}}{\design{#2}}}}

\DeclareFontFamily{OT1}{pzc}{}
\DeclareFontShape{OT1}{pzc}{m}{it}{<-> [1.1] pzcmi8t}{} 
\DeclareMathAlphabet{\mathpzc}{OT1}{pzc}{m}{it}
\newcommand{\pathLL}[1]{\mathpzc{#1}}
\newcommand{\view}[1]{\raisebox{.3ex}{$\ulcorner$}{#1}\raisebox{.3ex}{$\urcorner$}}
\newcommand{\fullview}[1]{\raisebox{.3ex}{$\ulcorner\mkern-6mu\ulcorner\mkern-2mu$}{#1}\raisebox{.3ex}{$\mkern-2mu\urcorner\mkern-6mu\urcorner$}}
\newcommand{\views}[1]{\view{#1}}
\newcommand{\fullviews}[1]{\fullview{#1}}
\newcommand{\PoD}[1]{{\mathcal{P}}_{#1}}

\newcommand{\dai}{\design{Dai}}

\usepackage{slashed}

\newlength{\dualwidth}
\newlength{\dualheight}
\newcommand{\dual}[2][1]{
\settowidth{\dualwidth}{$#2$}%
\settoheight{\dualheight}{$#2$}%
\makebox[\dualwidth][c]{\mbox{\rule{0cm}{#1\dualheight}$\Widetilde[#1]{#2}$}}
}

\newcommand{\PosOrder}{\olessthan}

\def\notPosOrderdef#1{\hbox{\hbox to -2pt{$#1/$\hss}$#1\PosOrder$}}

\newlength{\Overlineheight}
\newlength{\Overlinewidth}

\newcommand{\Overline}[2][\Overlinestretch]{%
\settowidth{\Overlinewidth}{$#2$}%
\settoheight{\Overlineheight}{$#2$}%
\setlength{\Overlineheight}{#1\Overlineheight}%
\rlap{$
\overline{
	\makebox[\Overlinewidth][c]{
		\rule{0cm}{\Overlineheight}
		}
	}$
}
#2
}

\newlength{\Widetildeheight}
\newlength{\Widetildewidth}

\newcommand{\Widetilde}[2][\Widetildestretch]{%
\settowidth{\Widetildewidth}{$#2$\hspace{1pt}}
\settoheight{\Widetildeheight}{$#2$}%
\setlength{\Widetildeheight}{#1\Widetildeheight}%
\rlap{$
\widetilde{
	\makebox[\Widetildewidth][c]{
		\rule{0cm}{\Widetildeheight} 
		}
	}$
}
#2
}

%% file: FIGURES/fig_NegJump.tex

\begin{figure}

\begin{center}
\scalebox{.9}{
\begin{tikzpicture}[>=latex]
\draw[->,black!40] (-7,0) -- (4.8,0);
\node at (5,0) {$\pathLL{p}$};

\node[black] at (-5.9,.4) {$\kappa'^-$};
\draw[black] (-6,-.1) -- (-6,.1);
\node at (-5.9,-.3) {$\alpha_n^-$};
\draw[black] (-5.6,-.1) -- (-5.6,.1);
\node at (-5.5,-.6) {$\alpha_n^+$};
\node at (-5,0.05) {\ldots};
\draw[black] (-4.5,-.1) -- (-4.5,.1);
\draw[<-,red] (-5.6,.1) .. controls (-5,.25) .. (-4.5,.1);
\node at (-4.4,-.3) {$\alpha_{n-1}^-$};

\node at (-3.35,0.15) {\ldots};
\node at (-3.35,-0.15) {\ldots};

\draw[black] (-2.3,-.1) -- (-2.3,.1);
\draw[black] (-1.9,-.1) -- (-1.9,.1);
\node at (-1.8,-.6) {$\alpha_{i+1}^+$};
\node at (-1.35,0.05) {\ldots};
\draw[black] (-.8,-.1) -- (-.8,.1);
\draw[<-,red] (-1.9,.1) .. controls (-1.35,.25) .. (-.8,.1);
\node at (-.7,-.4) {$\alpha_{i}^-$};
\draw[black] (-.4,-.1) -- (-.4,.1);
\node at (-.3,-.6) {$\alpha_{i}^+$};

\node at (.7,0.15) {\ldots};
\node at (.7,-0.15) {\ldots};

\draw[black] (1.8,-.1) -- (1.8,.1);
\draw[black] (2.2,-.1) -- (2.2,.1);
\node at (2.3,-.6) {$\alpha_{1}^+$};
\node at (2.75,0.05) {\ldots};
\draw[black] (3.3,-.1) -- (3.3,.1);
\draw[<-,red] (2.2,.1) .. controls (2.75,.25) .. (3.3,.1);
\node at (3.4,-.4) {$\alpha_{0}^-$};
\draw[black] (3.7,-.1) -- (3.7,.1);
\node at (3.8,-.6) {$\alpha_{0}^+$};
\node at (3.8,.4) {$\kappa^+$};

\draw[->,red] (3.55,.3) .. controls (-1.55,.95) .. (-5.86,.3);

\end{tikzpicture}
}
\end{center}
\caption{Path $\pathLL{p}$: Constraint of {\em negative jump} between $\kappa^+$ justified by $\kappa'^-$}
\label{fig:NegJump}
\end{figure}

%% file: FIGURES/fig_exempleNegJump.tex

\pgfdeclarelayer{background layer}
\pgfdeclarelayer{foreground layer}
\pgfsetlayers{background layer,main,foreground layer}

\begin{figure}

\begin{center}
\begin{minipage}{.40\textwidth}
\scalebox{.9}{
\begin{tikzpicture}[>=latex]

\node at (0,0) {$(+,\xi,\{1,2\})$};
\node at (-2,1) {$(-,\xi1,\{1\})$};
\node at (-2,2) {$(+,\xi11,\{0\})$};
\node at (-2,3) {$(-,\xi110,\{1\})$};
\node at (-2,4) {$(+,\xi1101,\{1\})$};
\node at (2,1) {$(-,\xi2,\{0,1\})$};
\node at (2,2) {$(+,\xi20,\{0\})$};
\node at (2,3) {$(-,\xi200,\{2\})$};
\node at (2,4) {$(+,\xi21,\{1\})$};

\draw (0,.25) -- (-2,.75);
\draw (-2,1.25) -- (-2,1.75);
\draw (-2,2.25) -- (-2,2.75);
\draw (-2,3.25) -- (-2,3.75);
\draw (0,.25) -- (2,.75);
\draw (2,1.25) -- (2,1.75);
\draw (2,2.25) -- (2,2.75);
\draw (2,3.25) -- (2,3.75);

\draw[->,blue,ultra thick,rounded corners] 
	(0,.2) -- (0,.3)
	-- (-.6,1) -- (-.6,2)
	-- (.6,1) -- (.6,2)
	-- (-.6,3) -- (-.6,4)
	-- (.6,3) -- (.6,4.4);

\node[anchor=base,blue] at (.4,4.5) {$\pathLL{p}$};

\draw[->,dashed,red,ultra thick] 
	(0,.2) .. controls(.7,.6) .. (.6,1)
	.. controls(.9,1.5) .. (.6,2)
	.. controls(.9,2.5) .. (.6,3)
	.. controls(.9,3.5) .. (.8,3.8)
	-- (.8,4.4);

\node[anchor=base,red] at (.95,4.5) {$\view{\pathLL{p}}$};

\end{tikzpicture}
}
\end{minipage}
\begin{minipage}{.40\textwidth}
\scalebox{.9}{
\begin{tikzpicture}[>=latex]

\node at (0,0) {$(+,\xi,\{1,2\})$};
\node at (-2,1) {$(-,\xi1,\{1\})$};
\node at (-2,2) {$(+,\xi11,\{0\})$};
\node at (-2,3) {$(-,\xi110,\{1\})$};
\node at (-2,4) {$(+,\xi1101,\{1\})$};
\node at (2,1) {$(-,\xi2,\{0,1\})$};
\node at (2,2) {$(+,\xi20,\{0\})$};
\node at (2,3) {$(-,\xi200,\{2\})$};
\node at (2,4) {$(+,\xi21,\{1\})$};

\draw (0,.25) -- (-2,.75);
\draw (-2,1.25) -- (-2,1.75);
\draw (-2,2.25) -- (-2,2.75);
\draw (-2,3.25) -- (-2,3.75);
\draw (0,.25) -- (2,.75);
\draw (2,1.25) -- (2,1.75);
\draw (2,2.25) -- (2,2.75);
\draw (2,3.25) -- (2,3.75);

\draw[->,blue,ultra thick,rounded corners] 
	(0,.2) -- (0,.3)
	-- (-.6,1) -- (-.6,2)
	-- (.6,1) -- (.6,2)
	-- (-.6,3) -- (.6,4)
	-- (.6,4.4);

\node[anchor=base,blue] at (.8,4.5) {$\pathLL{q}$};

\draw[->,dashed,red,ultra thick] 
	(0,.2) .. controls(-.7,.6) .. (-.6,1)
	.. controls(-.9,1.5) .. (-.6,2)
	.. controls(-.9,2.5) .. (-.6,3)
	.. controls(-.4,3.5) .. (.4,4)
	-- (.4,4.4);

\node[anchor=base,red] at (.25,4.5) {$\view{\pathLL{q}}$};
\end{tikzpicture}
}
\end{minipage}
\end{center}
\caption{The sequence $\pathLL{p}$ is a path, whereas $\pathLL{q}$ is not a path.}
\label{fig:ExNegJump}
\end{figure}

